%% file: newsubmission6.tex
\documentclass[11pt, twoside, letter]{article}
\usepackage[margin = 1in]{geometry}
\usepackage[utf8]{inputenc}
\usepackage{amsmath, amssymb, amsfonts, cite}
\usepackage{ifthen}
\usepackage{comment}

\usepackage{pgf}
\usepackage{tikz}
\usetikzlibrary{arrows,automata}

\usepackage{nicefrac}
\usepackage{xstring}

\usepackage{mathptmx}
\usepackage{tgtermes}
\usepackage[T1]{fontenc}
\usepackage{float}
\usepackage{microtype}

\usepackage{amsthm}
\usepackage[unicode,colorlinks=true,citecolor=blue,bookmarksnumbered=true,hyperfootnotes=false,linkcolor=green!40!black]{hyperref}
\usepackage{cleveref}

\newcommand{\qq}{\mbox{\boldmath $q$}}
\newcommand{\pp}{\mbox{\boldmath $p$}}

\usepackage{bm}
\renewcommand{\mathbf}[1]{\bm{#1}}

\usepackage{geometry}
\usepackage{fullpage}

\usepackage{pdfsync, graphicx}
\usepackage{marginnote}
\usepackage{framed}

\newcommand{\rr}[1]{\textcolor{black}{#1}}

\newtheorem{theorem}{Theorem}[section]

\newtheorem{claim}[theorem]{Claim}

\newtheorem{lemma}[theorem]{Lemma}
\newtheorem{corollary}[theorem]{Corollary}
\theoremstyle{definition}

\newtheorem{definition}{Definition}[section]

\renewcommand{\Pr}[2][]{\ensuremath{\mathbb{P}_{#1}\insq{#2}}}

\newcommand{\insq}[1]{\left[#1\right]}

\makeatletter
\newcommand*{\defeq}{\mathrel{\rlap{%
                     \raisebox{0.3ex}{$\m@th\cdot$}}%
                     \raisebox{-0.3ex}{$\m@th\cdot$}}%
                    =}
\newcommand*{\eqdef}{=
  \mathrel{\rlap{%
      \raisebox{0.3ex}{$\m@th\cdot$}}%
    \raisebox{-0.3ex}{$\m@th\cdot$}}%
}
\makeatother

\newcommand{\nfrac}[3][]{\nicefrac[#1]{#2}{#3}}

\newcommand{\etal}[0]{\emph{et al. }}

\renewcommand{\vec}[1]{\mathbf{#1}}
\newcommand{\xx}[0]{{\mathbf{x}}}
\newcommand{\yy}[0]{{\mathbf{y}}}
\newcommand{\ddelta}[0]{{\mathbf{\delta}}}
\newcommand{\CE}{\mbox{$\mathcal E$}}

\newcommand{\supp}{\mbox{Supp}}

\DeclareMathOperator*{\argmax}{arg\,max}

\begin{document}

\title{Mutation, Sexual Reproduction and Survival in Dynamic Environments}

\author{Ruta Mehta \\University of Illinois at Urbana-Champaign \\\texttt{ruta.mehta@gmail.com}
\and Ioannis Panageas \\ Georgia Institute of Technology \\\texttt{ioannis@gatech.edu}
\and Georgios Piliouras \\ Singapore University of Technology and Design \\\texttt{georgios.piliouras@gmail.com}
\and Prasad Tetali \\ Georgia Institute of Technology \\\texttt{tetali@math.gatech.edu}
\and Vijay V. Vazirani \\ Georgia Institute of Technology \\\texttt{vazirani@cc.gatech.edu}}

\date{}
\maketitle

\begin{abstract}
A new approach to understanding evolution \cite{DBLP:journals/jacm/Valiant09}, namely viewing it through the lens of computation,
has already started yielding new insights, e.g., natural selection under sexual reproduction can be interpreted 
as the {\em Multiplicative Weight Update} (MWU) Algorithm in coordination games played among genes \cite{PNAS2:Chastain16062014}. Using this machinery, we study the role of mutation in changing environments in the presence of sexual reproduction.
Following \cite{wolf2005diversity}, we model changing environments via a Markov chain, with the states representing environments,
each with its own fitness matrix. In this setting, we show that in the absence of
mutation, the population goes extinct, but in the presence of mutation, the population survives with positive probability. 

On the way to proving the above theorem, we need to establish some facts about dynamics in games. We provide the first, 
to our knowledge, polynomial convergence bound for noisy MWU in a coordination game.
Finally, we also show that in static environments, sexual evolution with mutation converges, for any level of mutation. 
\end{abstract}

\begin{figure}[hbtp]
   \centering
   \vspace{-0.95cm}
\includegraphics[width=75mm]{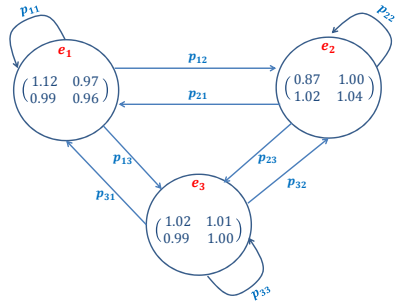}
   \vspace{-1cm}
\caption{An example of a Markov Chain model of fitness landscape evolution.}
\label{fig1}
\end{figure}

\thispagestyle{empty}
\newpage
\setcounter{page}{1}

\pagenumbering{arabic}

\section{Introduction}
Evolution has been the subject of intensive scientific investigation since the early 19th century,
yet many of its critical elements still remain under debate. 
A new, potent approach to studying evolution was initiated by Valiant \cite{DBLP:journals/jacm/Valiant09}, namely viewing it through the lens of computation.
This viewpoint has already started yielding concrete insights by translating qualitative hypotheses in biological systems to provable computational properties of Markov chains and other dynamical systems \cite{VishnoiOpen,Vishnoi:2013:MER:2422436.2422445, evolfocs14, PNAS2:Chastain16062014, ITCS15MPP, Meir15, PSV15},  which are standard hallmarks of TCS research.
We build on this direction whilst focusing on the challenge of evolving environments. 

Recent work due to Chastain, Livnat, Papadimitriou and Vazirani \cite{PNAS2:Chastain16062014} linked natural selection under sexual
reproduction ({\em sexual evolution}) to a
tangential field, namely dynamics in games. Building on the work of Nagylaki \cite{Nagylaki1}, they showed that this process can be interpreted
as the {\em Multiplicative Weight Update} (MWU) Algorithm \cite{Arora05themultiplicative}, which we call discrete replicator dynamics, in coordination games played among genes.  This connection opens up doors for applying tools from game theory and dynamical systems to understanding these fundamental processes.  Another important question is the
role of mutation, especially in the presence of changes to the environment. In the case of asexual reproduction, this was studied by
Wolf, Vazirani, and Arkin \cite{wolf2005diversity}. They modeled a changing environment via a Markov chain and described a model in which in the absence of
mutation, the population goes extinct, but in the presence of mutation, the population survives with positive probability. 

In this paper we study the next natural question, namely the role of mutation in the presence of sexual reproduction. The Chastain \etal 
result shows that MWU (sex) tries to maximize fitness as well as entropy via the process of recombination in genes. The question 
arises whether this is enough to safeguard against extinction in a changing environment, or is mutation still needed.

As in Chastain \etal \cite{PNAS2:Chastain16062014}, we will consider a haploid organism with two genes. 
Each gene can be viewed as a player in a game and the alleles of 
each gene represent strategies of that player. 
Once an allele is decided for each gene, an individual is defined, and its fitness is the payoff to each of the players, i.e.,
both players have the same payoff matrix, and therefore it is a coordination/partnership game. Each state of the
Markov chain represents an environment and has its own fitness matrix. 
We show the following under this model, where mutations are captured through a standard model appeared in \cite{Hofbauer98}:

\medskip
\noindent
{\bf Informal Theorem 1:} For a class of Markov chains (satisfying mild conditions), a haploid species
under {\em sexual evolution}\footnote{We refer to `evolution by natural selection under sexual reproduction' by {\em sexual evolution}
for brevity.} to without mutation dies out with probability one. In contrast, under sexual evolution \rm{with mutation} the
probability of long term survival is strictly positive.
\medskip

For each gene, if we think of its allele frequencies in a given population as defining a 
mixed strategy, then after reproduction, the frequencies change as per MWU \cite{PNAS2:Chastain16062014}. Furthermore, in the presence of 
mutation \cite{Hofbauer98}, every allele mutates to another allele of the corresponding gene in a small fraction of offsprings. 
As it turns out, in every generation the population size (of the species) changes by a multiplicative factor of the current expected payoff (mean
fitness). Hence, in order to prove Theorem 1, we need to analyze MWU (and its variant which captures mutations) in a time-evolving
coordination game whose matrix is changing as per a Markov chain.

The idea behind the first part of the theorem is as follows:
It is known that MWU converges, in the limit, to
a pure equilibrium in coordination games \cite{ITCS15MPP}. This implies that in a static environment, in the limit the population will be rendered monomorphic.
Showing such a convergence in a stochastically changing environment is not straightforward.
We first show that such an equilibrium can be reached \rm{fast enough} in a static environment. We then appeal to the
Borel-Cantelli theorem to argue that with probability one, the Markov chain will visit infinitely often and remain sufficiently long in one environment at some
point and hence the population will eventually become monomorphic. An assumption in our theorem is that for each individual, there
are bad environments, i.e., one in which it will go extinct. Eventually the monomorphic population will reach such an unfavorable
environment and will die out.

{\it Polynomial time convergence in static environment:}
For such a reasoning to be applicable we need a fast convergence result, which does not hold in the worst case, since 
by choosing initial conditions sufficiently close to the stable manifold of an unstable equilibrium, we are bound to spend
super-polynomial time near such unstable states. 
To circumvent this 
we take a typical approach of introducing a small noise into the dynamics \cite{pemantle90,Kleinberg09multiplicativeupdates,ge2015escaping}, 
and provide the first, to our knowledge, polynomial convergence bound for noisy MWU in
coordination games; this result is of independent interest. We note that pure MWU captures frequency changes of alleles in case of
infinite population, 
and the small noise can also be thought of as sampling error due to finiteness of the population. 
In the following theorem, dependence on all identified system parameters is necessary (see discussion in
Section \ref{sec:full}).  
\medskip

\noindent
{\bf Informal Theorem 2:} In static environments under small random noise ($||.||_{\infty}=\delta$), sexual evolution (without
mutation) converges with probability $1-\epsilon$ to a monomorphic fixed point in time $O\left(\frac{n\log
\frac{n}{\epsilon}}{\gamma^4\delta^6}\right)$, where $n$ is the number of alleles, and $\gamma$ the minimum fitness difference between
two genotypes.
\medskip

Although mutations seem to hurt mean population fitness in the short run in static environments, they are critical for
survival in dynamic environments, as shown in the second part of Theorem 1; it is proved as follows. 
The random exploration done by mutations and aided by the selection process, which
rapidly boosts the frequency of alleles with good mean fitness, helps the population survive. 
Essentially we couple the random variable capturing population size with a biased random walk, with a slight bias towards
increase. The result then follows using a well-known lemma on biased random walks.

{\it Robustness to mutations:} Finally we show that the convergence of MWU (without mutation) in static environments \cite{akin,ITCS15MPP}
can be extended to the case where mutations are also present. The former result critically hinges on
the fact that mean fitness strictly increases under MWU in coordination games, and thereby acts as a potential function.  
This is no more the case. However, using an inequality due to Baum and Eagon \cite{BE66} we manage to obtain a new potential function which is
the product of mean fitness and a term capturing diversity of the allele distribution. The latter term is essentially the product of allele
frequencies. 
\medskip
 
\noindent
{\bf Informal Theorem 3:} In static environments, sexual evolution with mutation converges, for any level of mutation. Specifically,
if we are not at equilibrium, at the next time generation at least one of mean population fitness or product of allele frequencies 
will increase. 
\medskip

Besides adding computational insights to biologically inspired themes, which to some extent may never be fully settled, we believe that our work is of interest even from a purely computational perspective. The nonlinear dynamical systems arising from these models are gradient-like systems of non-convex optimization problems. Their importance and the need to develop a theoretical understanding beyond worst case analysis has been pinpointed as a key challenge for numerous computational disciplines, e.g., from \cite{offconvex}:

\begin{quote}
\textit{``Many procedures in statistics, machine learning and nature at large -- Bayesian inference, deep learning, protein folding -- successfully solve non-convex problems \dots Can we develop a theory to resolve this mismatch between reality and the predictions of worst-case analysis?"}\end{quote}

Our theorems and techniques share this flavor. Theorem $1$ expresses time-average efficiency guarantees for gradient-like heuristics in the case of time-evolving optimization problems, Theorem $2$ argues about speedup effects by adding noise to escape out of saddle points, whereas Theorem $3$ is a step towards arguing about robustness to implementation details. We make this methodological similarities more precise by pointing them out in more detail in the related work section (see discussion in Section \ref{sec:related}). 
\medskip
\medskip

\noindent{\bf Organization of the paper.}
The rest of the paper is organized as follows. In the next section we discuss the relevant literature. 
In Section \ref{sec:prelims} we provide formal description of the model we analyze. 
In Section \ref{sec:overview}, we provide an overview of the proofs of our main theorems.
The proofs of main Theorems 1, 2, 3 and their formal statements appear in Section \ref{sec:survival}, \ref{sec:speed} and
\ref{sec:mutation} respectively.  
%
%
The omitted proofs can be found in the Appendix, along with explanation of the biological terms.

\input{relatedwork_pan3}

\input{prelims3}

\input{overview5}

\input{speed1}

\input{survival3}

\input{mutation1}


\input{tight}

\section{Conclusion and Open problems}\label{sec:conclusion}
In this paper we study various aspects of discrete replicator-like/MWUA dynamics and show three results: Two for dynamics with fixed
parameters, and one where the parameters evolve over time as per a Markov chain. Theorem
\ref{thm:convergencetime} establishes that a noisy version of discrete replicator dynamics converges {\em polynomially fast} to pure fixed
points in coordination games. Due to the connections established by Chastain \etal \cite{PNAS2:Chastain16062014}, this implies  that evolution under sexual
reproduction in haploids converges fast to a monomorphic population if the environment is static (fitness/payoff matrix is fixed).
Introducing mutations to this model, as in \cite{Hofbauer98}, augments the replicator dynamics, and our second result shows convergence for this
augmented replicator in coordination games. The proof is via a novel potential function, which is a combination of mean payoff and
entropy, which may be of independent interest.

Finally, for the replicator dynamics with noise, capturing finite populations, we show that assuming some mild
conditions, the population size will eventually become zero with probability one (extinction) under (standard) replicator, while under augmented
replicator (with mutations) it will never whither out (survival) with a non-trivial probability. 

\medskip
\noindent
A host of novel questions arise from this model and there is much space for future work.


\begin{itemize}
\item For the fast convergence result (first result above), we assumed that the random noise $\delta$ lies in a subset of hypercube of
length $\delta$, i.e., every entry $\delta_i$ is $\pm 1$ times magnitude $\delta$ and $\sum_i \delta_i=0$. Can the result be
generalized for a different class of random noise, where the noise also depends on the distribution of the alleles at every step and or
population size?  

\item The second result talks about convergence to fixed points, which happens at the limit (time $t \to \infty$). Therefore, an interesting
question would be to settle the speed of convergence. Additionally, for the {\em no mutations} case the result of \cite{ITCS15MPP}
shows that all the stable fixed points are pure. It would be interesting to perform stability analysis for the {\em replicator with mutations} as well.

\item Mutation can be modeled in an alternative way, where an individual can mutate to a completely new allele that is not part of some in advance fixed set of alleles. This is equivalent to adding a strategy to the coordination game. It will be interesting to define and analyze dynamics where mutation is modeled in such a way. Finally, what happens if environment changes are not completely independent but are instead affected by population size?
\end{itemize}

\bibliographystyle{alpha}
\bibliography{sigproc5,sigproc6,sigproc7}

\appendix

\input{appendix}

\newpage
\section{Figures}
To draw the phase portrait of a discrete time system $f:\Delta \to \Delta$, we draw vector $f(\vec{x})-\vec{x}$ at point $\vec{x}$. 

\begin{figure}[H]
   \centering
\includegraphics[width=75mm]{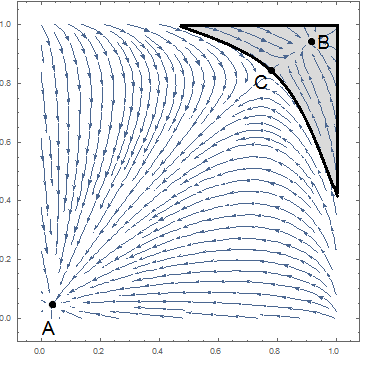}
\vspace{-0.6cm}
\caption{Example where population goes extinct in environment $e$ for some initial frequency vectors $(\vec{x},\vec{y})$ that are close to stable point $B$ (inside the shaded area). Mutation probability is $\tau = 0.03$ and the fitness matrix of environment $e$ is $W^e_{1,1} = 0.99, W^e_{2,2} = 2.09, W^e_{1,2}=0.37, W^e_{2,1}=0.56$. }
\label{fig2}
\end{figure}

\begin{figure}[H]
   \centering
\includegraphics[width=75mm]{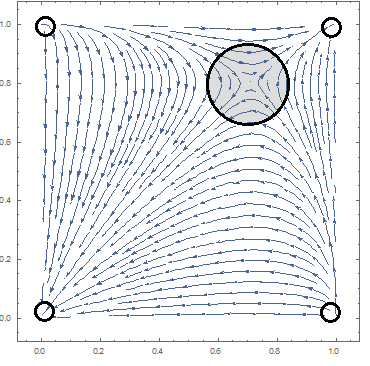}
\vspace{-0.5cm}
\caption{Example of dynamics without mutation in specific environment $W^e_{1,1} = 0.99, W^e_{2,2} = 2.09, W^e_{1,2}=0.37, W^e_{2,1}=0.56$. The circles qualitatively show all the points that slow down the increase in the average fitness $\vec{x}^TW^e\vec{y}$, i.e $\alpha$-close points or negligible.}
\label{fig3}
\end{figure}

\end{document}

%% file: relatedwork_pan3.tex
\section{Related Work}
\label{sec:related}



In the last  few years we have witnessed a rapid cascade of theoretical results on the intersection of computer science
and evolution. 
Livnat \etal \cite{PNAS1:Livnat16122008} introduced the notion of mixability, the ability of
an allele to combine itself successfully with others. 
In
\cite{ITCS:DBLP:dblp_conf/innovations/ChastainLPV13,PNAS2:Chastain16062014}  connections were established between sexual evolution and
 dynamics in coordination games. Meir and Parkes~\cite{Meir15} has provided a more detailed examination of these connections.
These dynamics are close variants of the standard (discrete) replicator dynamics~\cite{Hofbauer98}.
Replicator dynamics is closely connected to the multiplicative weights update algorithm~\cite{Kleinberg09multiplicativeupdates,2014arXiv1403.3885P}. In~\cite{ITCS15MPP} Mehta \etal established that these systems converge for almost all initial conditions to monomorphic states.
It is also possible to introduce connections between satisfiability and evolution~\cite{evolfocs14} as well as understand
the complexity of predicting the survival of diversity in complex species \cite{2014arXiv1411.6322M}.

The {\em error threshold} is the rate of errors in genetic mixing above which genetic information disappears~\cite{Eig93}.
Vishnoi~\cite{Vishnoi:2013:MER:2422436.2422445} showed existence of such sharp thresholds. Moreover, in \cite{PSV15}
Panageas \etal shed light on the speed of asexual evolution (see also \cite{VishnoiSpeed}).
Finally, in~\cite{DSV} Dixit \etal present finite
population models for asexual haploid evolution that closely track the standard infinite population model of 
Eigen~\cite{Eigen71}. 
Wolf, Vazirani, and Arkin \cite{wolf2005diversity} analyzed models of mutation and survival of diversity also for asexual populations
but the dynamical systems in this case are linear and the involved methodologies are rather different.

Introducing noise in non-linear dynamics  have been shown to be able to simplify the analysis of nonlinear dynamical systems by
``destroying'' Turing-completeness of classes of dynamical systems and thus making the system's long-term behavior computationally
predictable \cite{Braverman:2012:NVC:2090236.2090247}. Those techniques focus on establishing invariant measures for the systems of
interest and computing their statistical characteristics. In our case, our unperturbed dynamical systems have exponentially many saddle
points and numerous stable fixed points and species survival is critically dependent on the the amount of time that trajectories spend
in the vicinity of these points thus much stronger topological characterizations  are necessary. Adding noise to game theoretic
dynamics \cite{Kleinberg09multiplicativeupdates,ackermann2009concurrent,chien2007convergence} to speed up convergence to approximate
equilibria in potential games is a commonly used approach in algorithmic game theory, however, the respective proof techniques and
notions of approximation are typically sensitive to the underlying dynamic, the nature of noise added as well as the details of the
class of games.  
 
In the last year there has been a stream of work on understanding how gradient (and more generally gradient-like) systems escape out of
the saddle fixed points fast \cite{ge2015escaping,lee2016gradient}. This is critically important for a number of computer science
applications, including speeding up the training of deep learning networks. The approach pursued by these papers is similar to our
own, including past papers in the line of TCS and biology/game theory literature
\cite{Kleinberg09multiplicativeupdates,2014arXiv1403.3885P,ITCS15MPP}. For example, in \cite{2014arXiv1403.3885P,ITCS15MPP} it has been established that in
non-convex optimization settings gradient-like systems (e.g., variants of Multiplicative Weights Updates Algorithm) converge for all
but a zero measure of initial conditions to local minima of the fitness landscape (instead of saddle points even in the presence of
exponentially many saddle points). Moreover, as shown in  \cite{Kleinberg09multiplicativeupdates} noisy dynamics diverge fast from the
set of saddle points whose Jacobian has eigenvalues with large positive real parts. Similar techniques and arguments can be applied to argue generic convergence to
local minima of numerous other dynamics (including noisy/deterministic versions of gradient dynamics).  Finally, using techniques
developed in \cite{2014arXiv1403.3885P}, one can argue that gradient dynamics converge to local minima with probability one in
non-convex optimization problems even in the presence of continuums of saddle points \cite{PP16}, answering an open question in
\cite{lee2016gradient}. We similarly hope that techniques developed here about fast and robust convergence can also be extended to other classes of gradient(-like) dynamics in non-convex optimization settings.
 
Finite population evolutionary models over time evolving fitness landscapes are typically studied via simulations (e.g., \cite{liekens2005evolution} and references therein). These models have also inspired evolutionary models of computation, e.g., genetic algorithms, whose study under dynamic fitness environments is a well established area with many applications (e.g., \cite{yang2007evolutionary} and references therein) but with little theoretical understanding and even theoretical papers on the subject typically rely on combinations of analytical and experimental results  \cite{Branke2003}.

%% file: prelims3.tex
\section{Preliminaries}
\label{sec:prelims}

\noindent{\bf Notation:} All vectors are in bold-face letters, and are considered as column vectors. To denote a row vector we use
$\xx^T$. The $i-{th}$ coordinate of $\xx$ is denoted by $x_i$. Let $\Delta_n=\{\xx\in \mathbb R^n\ |\ \xx\ge 0,\ \sum_{i=1}^n x_i=1\}$
be the set of probability distributions on $n$ coordinates. For given matrix $A$ define $A_{\max}, A_{\min}$ the largest, smallest
entry in matrix $A$ respectively and $(A\vec{x})_i \defeq \sum_{j}A_{ij}x_{j}$. 
Define $\supp(\xx)=\{i\ |\ x_i\neq 0\}$.  

\subsection{Dynamics: Discrete Replicator Dynamics with/without Mutation}\label{subsec:naturalsel}
For a haploid species\footnote{See Section \ref{asec.bioTerms} in the appendix for a short discussion of all relevant biological terms.} (one with single set of chromosomes, unlike diploids such as humans who have chromosome pairs) with
two genes (coordinates), let $S_1$ and $S_2$ be the set of possible alleles (types) for the first and second gene respectively. Then, an individual of such a
species can be represented by an ordered pair $(i,j)\in S_1 \times S_2$. Let $W_{ij}$ be the fitness
of such an individual capturing its ability to reproduce during a mating. Thus, fitness landscape of such a species can be represented
by matrix $W$ of dimension $n\times n$, where we assume that $n=|S_1|=|S_2|$. 
\medskip

\textbf{Sexual Model without mutation:} 
In every generation, each individual $(i,j)$ mates with another
individual $(i',j')$ picked uniformly at random from the population (can pick itself). The offspring can have any of the
four possible combinations, namely $(i,j), (i,j'),(i',j),(i',j')$, with equal probability. Let $x_i$ be a random variable that denotes
the proportion of the population with allele $i$ in the first coordinate, and similarly $y_j$ be the frequency of the population with allele
$j$ in the second coordinate. After one generation, the expected number of offsprings with allele $i$ in first coordinate is proportional to
$x_i \cdot x_i \cdot (W\vec{y})_i + 2\frac{1}{2}(1-x_i)x_i \cdot (W\vec{y})_i = x_i (W\vec{y})_i$ ($x_i^2$ stands for the
probability both individuals have allele $i$ in the first coordinate - which the offspring will inherit - and $2\frac{1}{2}(1-x_i)x_i$
stands for the probability that exactly one of the individuals has allele $i$ in the first coordinate and the offspring will inherit).
Similarly the expected number of offsprings with allele $j$ for the second coordinate is $y_j (W^T\vec{x})_j$. Hence, if
$\vec{x}',\vec{y}'$ denote the frequencies of the alleles in the population in the next generation (random variables) 
$$E[x'_i|\vec{x},\vec{y}] = \frac{x_i (W\vec{y})_i}{\vec{x}^TW\vec{y}} \textrm{ and
}E[y'_j|\vec{x},\vec{y}] = \frac{y_j (W^T\vec{x})_j}{\vec{x}^TW\vec{y}}.$$

We are interested in analyzing a \textit{deterministic} version of the equations above, which essentially captures an infinite
population model. Thus if frequencies at time $t$ are denoted by $(\vec{x}(t),\vec{y}(t))$, they obey the following 
dynamics governed by the function $g:\Delta \rightarrow \Delta$, where $\Delta = \Delta_n \times \Delta_n$:

\begin{equation}\label{eq.g}
\mbox{Let $(\xx(t+1),\yy(t+1)) = g(\xx(t),\yy(t))$, where}
\begin{array}{c}
\forall i\in S_1,\ x_i(t+1) = x_i(t) \frac{(W\vec{y}(t))_i}{\vec{x}^T(t)W\vec{y}(t)}\\
\forall j\in S_2,\ y_j(t+1) = y_j(t) \frac{(W^T\vec{x}(t))_j}{\vec{x}^T(t)W\vec{y}(t)}.
\end{array}
\end{equation}

It is easy to see that $g$ is well-defined when $W$ is a positive matrix. Chastain \etal \cite{PNAS2:Chastain16062014} gave a game
theoretic interpretation of the deterministic equations (\ref{eq.g}). It can be seen as a repeated two player coordination game (each
gene is a player), the possible alleles for a gene are its pure strategies and both players play according to dynamics (\ref{eq.g}). A
modification of these dynamics has also appeared in models of grammar acquisition \cite{novak2001evolgrammar}, and can be seen as the
discrete analogue of continuous replicator dynamics \cite{akin}. Furthermore, Mehta \etal  \cite{ITCS15MPP} showed that dynamics with
equations (\ref{eq.g}) converges point-wise to a pure fixed point, { i.e.}, where exactly one coordinate is non-zero in both $\xx$
and $\yy$, for all but measure zero of initial conditions in $\Delta$, when $W$ has distinct entries.
\medskip

\noindent
\textbf{Sexual Model with mutation:} Next we extend the dynamics of (\ref{eq.g}) to incorporate mutation.
The mutation model which appears in Hofbauer's book \cite{Hofbauer98}, is a two step process. The first step is governed by
(\ref{eq.g}), and after that in each individual, and for each of its gene, corresponding allele, say $k$, mutates to another allele of the same
gene, say $k'$, with probability $\tau>0$ for all $k'\neq k$. 
%
After a simple calculation (see \ref{app:calc} for calculations) the resulting dynamics turns out to be as follows, where $f$ is a 
$\Delta \rightarrow \Delta$ function:

\begin{equation}\label{eq.fs}
\mbox{Let $(\vec{x}(t+1),\vec{y}(t+1))=f(\vec{x}(t),\vec{y}(t))$, then }
\begin{array}{lcl}
x_i(t+1) & = & (1-n\tau) x_i(t) \frac{(W\vec{y}(t))_i}{\vec{x}(t)^TW\vec{y}(t)}+\tau,\ \forall i\in S_1 \\
y_j(t+1) & = & (1-n\tau) y_j(t) \frac{(\vec{x}(t)^TW)_j}{\vec{x}(t)^TW\vec{y}(t)} + \tau,\ \forall j\in S_2.
\end{array}
\end{equation}


\subsection{Our model}\label{subsec:finitepopulation}
In this paper we will analyze a noisy version of (\ref{eq.g}), (\ref{eq.fs}). 
Essentially we add small random noise to non-zero coordinates of $(\vec{x}(t), \vec{y}(t))$ \footnote{This is different from diffusion approximation, noise helps to avoid saddle points essentially.}. 

\begin{definition}\label{def:noise}

Given $\vec{z} \in \Delta$ and a small $0<\delta$ ($\delta$ is $o_n(\tau)$), define $\Delta(\vec{z},\delta)$ to be a set of vectors
$\{\vec{z}+\ddelta \in \Delta\ | \ \supp(\ddelta) = \supp(\vec{z});\ \delta_{i} \in \{-\delta, +\delta\},\ \forall i\}$.\footnote{In
case the size of the support of $\vec{z}$ is odd, there will be a zero entry in $\vec{\delta}$, so $|\supp(\vec{\delta})| =
|\supp(\vec{z})|-1$}
\end{definition}

Note that if $\vec{z}$ is pure (has support size one), then $\ddelta$ is all zero vector\footnote{There are no sampling errors in
monomorphic population}. Define noisy versions of both $g$ from
(\ref{eq.g}) and $f$ from (\ref{eq.fs}) as follows: Given $(\xx(t),\yy(t))$ pick $\ddelta_{\vec{x}}\in \Delta(\xx(t),\delta)$ and
$\ddelta_{\vec{y}} \in \Delta(\yy(t),\delta)$ uniformly at random. Set with probability half $\ddelta_{\vec{x}}$ to zero,
and with the other half set $\ddelta_{\vec{y}}$ to zero. 
Then redefine dynamics $g$ of (\ref{eq.g}) as follows: 

\begin{equation}\label{eq.gn}
(\xx(t+1),\yy(t+1))=g_\delta(\xx(t),\yy(t)) = g(\xx(t),\yy(t))+(\ddelta_{\xx},\ddelta_{\yy}).
\end{equation}

\noindent And redefine dynamics $f$ of (\ref{eq.f}) capturing sexual evolution with mutation as follows. 

\begin{equation}\label{eq.fn}
(\xx(t+1),\yy(t+1))=f_\delta(\xx(t),\yy(t)) = f(\xx(t),\yy(t))+(\ddelta_{\xx},\ddelta_{\yy}).
\end{equation}

Furthermore, we will have that if any $x_i,y_j$ goes below $\delta$, we set it to zero. 
 This is crucial for our
theorems because otherwise the dynamics with equations (\ref{eq.g}) and (\ref{eq.fs}) (or even (\ref{eq.gn}) and (\ref{eq.fn})) can
converge to a fixed point at $t \to \infty $, but never reach a point in a finite amount of time. This is true in the result of
\cite{ITCS15MPP}, the dynamics converge almost surely to pure fixed points as $t \to \infty$ but do not reach fixation in a finite
time. So $x_i,y_j$ reaches fixation (set it to zero) if $x_i,y_j<\delta$. We need to re-normalize after this step. 

\begin{equation}\label{eq.zero}
\begin{array}{c}
\mbox{$\forall i\in S_1$, if $x_i(t)<\delta$ then set $x_i(t)=0$. Re-normalize $\xx(t)$.}\\
\mbox{$\forall j\in S_1$, if $y_j(t)<\delta$ then set $y_j(t)=0$. Re-normalize $\yy(t)$.}
\end{array}
\end{equation}
\begin{definition}\label{def:negligible}
We call a vector $\vec{v}$ {\em negligible} if there exists an $i$ s.t $v_i < \delta$. 
\end{definition}

\noindent{\bf Tracking population size.} Suppose the size of the initial population is $N^0$, and let population at time $t$ be $N^t$.
In every time period $N^t$ gets multiplied by average fitness of the current population, namely $\vec{x}(t) ^{ T}W(t)\vec{y}(t)$, where
$(\vec{x}(t),\vec{y}(t))$ denote the frequencies of alleles at generation $t$ and $W(t)$ the matrix fitness/environment at time (see discussion below about changing of environments).

\begin{equation}\label{eq.fit}
\mbox{Let average fitness $\Phi^t=\vec{x}(t)^TW(t)\vec{y}(t)$} \\
\mbox{ then $\mathbb{E}[N^{t+1}| \vec{x}(t),\vec{y}(t), N^t] = N^t  \Phi^{t+1}$}
\end{equation}
We will consider $N^{t+1} = N^t  \Phi^{t+1}$
(see also \cite{populationsize}). Based on the value of $N^t$, we give the definition of survival and extinction.

\begin{definition}\label{def.survive}
We say the population {\em goes extinct} if for initial population size $N^0$, there exists a time $t$ so that $N^t
< 1$. On the other hand, we say that population {\em survives} if for all times $t \in \mathbb{N}$ we have that $N^t \geq 1$.
\end{definition}

\subsubsection{Model of environment change}\label{sec:envchange}
Following the work of Wolf \etal \cite{wolf2005diversity}, we consider a Markov chain based model of changing environment. Let $\CE$ be the set of
different possible environments, and $W^e$ be the fitness matrix in environment $e\in \CE$. $E$ denotes the set of $(e,e')$
pairs if there is a non-zero probability $p_{e,e'}\in (0, 1)$ to go from environment $e$ to $e'$. See Figure \ref{fig1} for
an example.
For a parameter $p<1$ we assume that $\sum_{e': (e,e')\in E} p_{e,e'}\le p,\ \forall e \in \CE$. That is, after every generation of the dynamics (\ref{eq.gn}) or (\ref{eq.fn}), the environment changes to one of its neighboring environment with probability
at most $p<1$, and remains unchanged with probability at least $(1-p)$. The graph formed by edges in $\CE$ is
assumed to be connected, thus the resulting Markov chain eventually will stabilize to a stationary distribution $\pi_e$ (is ergodic).


Even though fitness matrices $W^e$ can be arbitrary, it is generally assumed that $W^e$ has distinct positive
entries \cite{ITCS:DBLP:dblp_conf/innovations/ChastainLPV13,ITCS15MPP}.
Furthermore, no individual can survive all the environments on average. Mathematically, if $\pi_e$ is the stationary distribution
of this Markov chain then,
$\forall i,j,\ \ \ \prod_{e \in \CE} (W_{ij}^e)^{\pi_e} <1$.
Furthermore, we assume that every environment has alleles of good type as well as bad type. 
An allele $i$ of good type has uniform fitness (i.e., $\frac{\sum_j W_{ij}}{n}$) of at least $(1+\beta)$ for some $\beta>0$,
and alleles of bad type are dominated by a good type allele point-wise.
\footnote{Think of bad type alleles akin to a terminal genetic illness. Such assumptions are typical in the biological literature (e.g., \cite{liekens2005evolution}).}
Finally, the number of bad alleles are $o(n)$ (sublinear in $n$).
Let the set of bad alleles for genes $i=1,2$ in environment $e$ be denoted by $B^e_i$.
%

Putting all of the above together, the Markov chain for environment change is defined by set $\CE$ of environments and its adjacency
graph, fitness matrices $W^e$, $\forall e \in \CE$, probability $1-p$ with which dynamics remains in current environment, sets
$B^e_i\subset S_i,\ i=1,2$ of bad alleles in environment $e$, and $\beta>0$ to lower-bound average fitness of good type alleles. See also
Section \ref{ass:environment} for discussion on the assumptions where we claim that most of them are necessary for our theorems. In the next
sections we will analyze the dynamics with equations (\ref{eq.gn}, \ref{eq.fn}) in terms of convergence and population size for fixed and
dynamic environments.


\begin{table}[ht!]
     \centering

\begin{tabular}{ |p{3cm}|p{13cm}|  }
 \hline
 Symbol & Interpretation\\
 \hline
  \hline
$W^e$ & fitness matrix at environment $e$\\
\hline
$W(t), W^{e(t)}$ & fitness matrix at time $t$\\
\hline
$\gamma^e$ & minimum difference between entries in fitness matrix $W^{e}$ \\
  \hline
 $\vec{x}, \vec{y}$ & frequencies of (alleles) strategies \\
  \hline
 $\vec{\delta}$ & noise/perturbation \\
 \hline
   $\Phi$ & potential/average fitness $\vec{x}^TW\vec{y}$\\
  \hline
$\beta$ &  If allele $i$ is of good type in environment $e$ then it satisfies $\frac{\sum_j W^e_{ij}}{n}\geq 1+\beta$ \\
 \hline
 $\tau$ & probability that an individual with allele $k$ mutates to $k'$ (of the same gene) \\
 \hline
\end{tabular}
 \caption{List of parameters}

\end{table}

%% file: overview5.tex
\section{Overview of proofs}
\label{sec:overview}

\rr{The dynamical systems that we analyze, namely (\ref{eq.gn}) and (\ref{eq.fn}), under the evolving environment model of Section
\ref{sec:envchange} are}
 (stochastically perturbed) nonlinear replicator-like dynamical systems
whose parameters evolve according to a (possibly slow mixing) Markov chain. We reduce the analysis of this complex setting to a series of smaller, modular arguments that combine as set-pieces to produce our main theorems.


\smallskip

\noindent{\em Convergence rate for evolution without mutation in static environment.}
Our starting point is \cite{ITCS15MPP} where it was shown that in the case of noise-free sexual dynamics \rr{governed by (\ref{eq.g})}
the average population fitness increases in each step and the system converges to equilibria, and moreover that for almost all initial
conditions the resulting fixed
point corresponds to a monomorphic population (pure/not mixed equilibrium). Conceptually, the first step in our analysis tries to
capitalize on this stronger characterization by showing that convergence to such states happens fast. This is critical because while
there are only linearly many pure equilibria, there are (generically) exponentially many isolated, mixed ones
\cite{ITCS:DBLP:dblp_conf/innovations/ChastainLPV13}, which are impossible to meaningfully characterize. By establishing the predictive
power of pure states, we radically reduce our uncertainty about system behavior and produce a building block for future arguments.


Without noise we cannot hope to prove fast convergence to pure states since by choosing initial conditions sufficiently close to the
stable manifold of an unstable equilibrium, we are bound to spend  super-polynomial time near such unstable states. In finite
population models, however,  the system state (proportions of different alleles) is always subject to small stochastic shocks (akin to
sampling errors). These small shocks suffice to argue fast convergence by combining an inductive argument and a potential/Lyapunov
function argument. 

To bound  the convergence time to a pure fixed point starting at an arbitrary mixed strategy (maybe with full support), it suffices to
bound the time it takes to reduce the size of the support by one, because once a strategy \rr{$x_i$} becomes zero it remains zero \rr{under
(\ref{eq.gn})}, { i.e.,} an
extinct allele can never come back in absence of mutations (and then use induction). For the inductive step, we need two non trivial
arguments. First we need a lower bound on the rate of increase of the mean population fitness when the dynamics is not at approximate
fixed points\footnote{We call these states $\alpha$-close points.}, shown in Lemma \ref{lem:inequality}. This requires
a quantitative strengthening of potential/(nonlinear dynamical system) arguments in \cite{ITCS15MPP}. Secondly, we
show that the noise suffices to escape fast (with high probability) from the influence of fixed points that are not monomorphic (these are like saddle points).  This requires a
combination of stochastic techniques including origin returning random walks,
Azuma type inequalities for submartingales, and arguing about the increase in expected mean fitness \rr{$\xx(t)^TW(t)\yy(t)$} in a few steps (Lemmas
\ref{lem:zeroinexpectation}-{\ref{lem:submartingale}), \rr{where $\xx$ and $\yy$ capture allele frequencies at time step $t$}. \rr{As a
result we show polynomial-time convergence of (\ref{eq.gn}) to pure equilibrium under static environment in Theorem
\ref{thm:convergencetime}. This result may be of independent interest since fast convergence of nonlinear dynamics to equilibrium is not typical \cite{Meiss2007}. }

  \smallskip

\noindent{\em Survival, extinction under dynamic environments}. \rr{As described in Section \ref{sec:envchange},} we consider a Markov chain
based model of environmental changes, where after every selection step, the fitness matrix changes with probability at most $p$.
Suppose the starting population size is $N^0>0$ and let $N^t$ denote the size at time $t$ then in every step $N^t$ gets multiplied by the
mean fitness $\xx(t)^TW(t)\yy(t)$ of the current population (see (\ref{eq.fit})). We say that population goes extinct if for some $t$,
$N^t<1$, and it survives if $N^t\ge 1$, for all $t$. 

We assume that there do not exist "all-weather" phenotypes. 
We encode this by having the monomorphic population of any genotype decrease when matched to an environment chosen according to the
stationary distribution of the Markov chain.\footnote{\rr{If, for any genotype, the population increased in expectation over the randomly
chosen environment, then once monomorphic population consisting of only such a genotype is reached, the population would blow up
exponentially (and forever) as soon as the Markov chain reached its mixing time.}} In other words, an allele may be both ``good'' and
``bad'' as environment changes, sometimes leading to growth, and other times to decrease in population. 

\textit{Case a) sexual evolution without mutation:} If the population becomes monomorphic then this single phenotype can not survive in
all environments, and will eventually wither as its population will be in exponential decline once the Markov chain mixes. The question
is whether monomorphism is achieved under changing environment; the above analysis is not applicable directly as the fitness matrix is not
fixed any more. Our first theorem \rr{(Theorem \ref{thm:convergencetime})} upper bounds the amount of time $T$ needed to ``wait'' in a
single environment so as the probability of convergence to a monomorphic state is at least some constant ( e.g., $\frac{1}{2}$).
Breaking up the time history in consecutive chunks of size $T$ and applying Borel-Cantelli theorem implies that the population will
become monomorphic with probability one (Theorem \ref{thm:nomutationdie}). This is the strongest possible result without explicit
knowledge of the specifics of the Markov chain (e.g., mixing time).

\textit{Case b) sexual evolution with mutation:}
\rr{As described in Section \ref{subsec:naturalsel}}, we consider a well-established model of mutation \cite{Hofbauer98}, where after
every selection step, each allele mutates to another with probability $\tau$. \rr{The resulting dynamics is
governed by (\ref{eq.fs}), and we analyze its noisy counterpart (\ref{eq.fn})}.
This ensures that in each period the proportion of every allele is at least $\tau$. We show that this helps the population survive.

Unlike the {\em no mutation} case \cite{ITCS15MPP}, the average fitness \rr{$\xx(t)^TW\yy(t)$} is no more increasing in every step, even
in absence of noise. Instead we derive another potential function that is a combination of average fitness and entropy.  Due to
mutations forcing exploration, natural selection weeds out the bad alleles fast (Lemma \ref{lem:badallele}). Thus there may be initial
decrease in fitness, however the
decrease is upper bounded. Furthermore, we show that the fitness is bound to increase significantly within a short time horizon due to
increase in population of good alleles (Lemma \ref{lem:phasetransition}). Since population size
gets multiplied by average fitness in each iteration, this defines a biased random walk on logarithm of the population size. Using
upper and lower bounds on decrease and increase respectively, we show that the probability of extinction stochastically dominates a
simpler to analyze random variable pertaining to biased random walks on the real line (Lemma \ref{lem:biased}). Thus, the probability of long term
survival is strictly positive (Theorem \ref{thm:mutationsurvive}). This completes the proof of informal Theorem 1. 

\textit{Deterministic convergence despite mutation in static environments:} Finally, as an independent result for the case of noise free
dynamics (infinite population) with mutation \rr{governed by (\ref{eq.fs})}, we show convergence to fixed points in the limit, by defining a novel
potential function which is the product of mean fitness $\xx^T W\yy$ and a term capturing diversity of the allele distribution (Theorem
\ref{thm:potential}).  The latter term is essentially the product of allele frequencies \rr{($\prod_{i}x_{i}\prod_{i}y_{i}$)}. 
Such convergence results are not typical in dynamical systems literature \cite{Meiss2007}, and therefore this potential function may
be useful to understand limit points of this and similar dynamics (the continuous time analogue can be found here
\cite{Hofbauer98}). One way to interpret this result is a homotopy method for computing equilibria in coordination games, where the
algorithm always converges to fixed points, and as mutation goes to zero the stable fixed points correspond to the pure Nash equilibria
\cite{ITCS:DBLP:dblp_conf/innovations/ChastainLPV13}.

%% file: speed1.tex
\section{Rate of Convergence: Dynamics without Mutation in Fixed Environments}\label{sec:speed}
In this section we show a polynomial bound on the convergence time of dynamics (\ref{eq.gn}), governing sexual evolution under natural
selection with noise, in a static environment. In addition, we show that the fixed points reached by the dynamics are pure. 

Consider a fixed environment $e$ and we use $W$ to denote its fitness matrix $W^e$.
It is known that average fitness $\vec{x}^TW\vec{y}$ increases under the non-noisy counterpart (\ref{eq.g}) \cite{ITCS15MPP}. In the next lemma
we obtain a lower bound on this increase. 

\begin{lemma} \label{lem:inequality} Let $(\hat{\vec{x}},\hat{\vec{y}}) = g(\vec{x},\vec{y})$ where $(\vec{x},\vec{y}) \in \Delta$ and
$g$ is from equation (\ref{eq.g}). Then, $$\hat{\vec{x}}^TW\hat{\vec{y}} - \vec{x}^TW\vec{y} \geq C\left(\sum_{i}x_{i} \left((W\vec{y})_{i} -
\vec{x}^TW\vec{y}\right)^2+\sum_{i}y_{i} \left((W^T\vec{x})_{i} - \vec{x}^TW\vec{y}\right)^2\right )$$ for
$C=\frac{3}{8 \cdot \max_{i,j}W_{ij}}$.
\end{lemma}
For the rest of the section, $C$ denotes $\frac{3}{8 \cdot W_{\max}}$ where $W_{\max} = \max_{ij}W_{ij}$ and $W_{\min} =
\min_{ij}W_{ij}$. Note that the lower bound obtained in Lemma \ref{lem:inequality} is strictly positive unless $(\xx,\yy)$ is a
fixed point of (\ref{eq.g}). This gives an alternate proof of the fact that, under dynamics (\ref{eq.g}), average fitness is a
potential function, { i.e.,} increases in every step. On the other hand, the lower bound can be arbitrarily small at some points, and
therefore it does not suffice to bound the convergence time. Next, we define points where this lower-bound is relatively small.



\begin{definition} We call a point $(\vec{x},\vec{y})$ $\alpha$-close for an $\alpha>0$, if for
all $\vec{x}', \vec{y}' \in \Delta$ such that $\supp(\vec{x}') \subseteq \supp(\vec{x})$ and $\supp(\vec{y}')\subseteq\supp(\vec{y})$
we have $|\vec{x}^TW\vec{y} - \vec{x}'^TW\vec{y}| \le \alpha$ and $|\vec{x}^TW\vec{y} - \vec{x}^T W\vec{y}'|\le \alpha$.
\end{definition}

$\alpha$-close points, are a specific class of  {\em ``approximate'' stationary points}, where the progress in average fitness is not significant (see
Figure \ref{fig3}, the big circles contain these points).  From now on, think $\alpha$ as a small parameter that will be determined in
the end of this section. If a given point $(\vec{x},\vec{y})$ is not $\alpha$-close and not {\em negligible} (see Definition
\ref{def:negligible}) then using Lemma \ref{lem:inequality} it follows that the increase in potential is at least $C \delta \alpha^2$.
Formally:
\begin{corollary}\label{cor:notaclose}
If $(\vec{x},\vec{y})\in \Delta$ is neither $\alpha$-close nor negligible, and $(\hat{\xx},\hat{\yy})=g(\xx,\yy)$, then
$$\hat{\vec{x}}^TW\hat{\vec{y}} \geq \vec{x}^TW\vec{y}+C\delta \alpha^2.$$
\end{corollary}
\begin{proof} Since the vector $(\vec{x},\vec{y})$ is neither $\alpha$-close nor negligible, it follows that there exists an index $i$
such that $|(W\vec{y})_i-\vec{x}^TW\vec{y}|>\alpha$ and $x_i \geq \delta$ and hence $x_i ((W\vec{y})_i-\vec{x}^TW\vec{y})^2 > \delta
\alpha^2$, or $|(W^T\vec{x})_i-\vec{x}^TW\vec{y}|>\alpha$ and $y_i\geq\delta$ and hence $y_i ((W^T\vec{x})_i-\vec{x}^TW\vec{y})^2 >
\delta \alpha^2$. Therefore in Lemma \ref{lem:inequality}, the r.h.s is at least $C\delta \alpha^2$ and thus we get that
$\hat{\vec{x}}^TW\hat{\vec{y}} - \vec{x}^TW\vec{y} \geq C\delta \alpha^2.$
\end{proof}

In the analysis above we considered non-noisy dynamics governed by (\ref{eq.g}). Our goal is to analyze finite population dynamics,
which introduces noise and the resulting dynamics is (\ref{eq.gn}).  This changes how the fitness increases/decreases. The next lemma shows
that in expectation the average fitness remains unchanged after the introduction of  noise.

\begin{lemma}\label{lem:zeroinexpectation} Let $\vec{\delta}= (\vec{\delta}_{\vec{x}},\vec{\delta}_{\vec{y}})$ be the  noise vector. It
holds that $\mathbb{E}_{\vec{\delta}}[(\vec{x+\vec{\delta}_{\vec{x}}})^TW(\vec{y}+\vec{\delta}_{\vec{y}})] = \vec{x}^TW\vec{y}$.
\end{lemma}
Next, we show how random noise can help the dynamic escape from a polytope of $\alpha$-close points.  We first analyze how adding noise may help
increase fitness with high enough probability. A simple application of Catalan numbers shows that:

\begin{lemma} \label{lem:ballot} The probability of a (unbiased) random
walk on the integers that consist of $2m$ steps of unit length, beginning at the origin and ending at the origin, that never becomes
negative is $\frac{1}{m+1}$.
\end{lemma}

\noindent
We define $\gamma = \min_{(i,j) \neq (i',j')}  |W_{ij} - W_{i'j'}|$.
The following lemma is essentially a corollary of Lemma \ref{lem:ballot}.

\begin{lemma}\label{lem:goodjump} Let $\vec{\delta}_{\vec{y}}$ be a random noise with support size $m$. For all $i$ in the support of
${\vec{x}}$ we have that $(W\vec{\delta}_{\vec{y}})_i \geq \frac{\gamma \delta m}{2}$ with probability at least
$\frac{1}{1+m/2}$ (same is true for $\vec{\delta}_{\vec{x}}$ and $\vec{y}$).  \end{lemma}

\noindent
We will also need the following theorem due to Azuma \cite{dubhashi09} on submartingales.

\begin{theorem}\label{thm:azuma} [Azuma \cite{dubhashi09}]. Suppose $\{ X_k, k=0,1,2,...,N\}$ is a submartingale and also $|X_k -
X_{k-1}|<c$ almost surely then for all positive integers $N$ and all $t>0$ we have that $$\Pr{X_N-X_0 \leq -t} \leq e^{-\frac{t^2}{2Nc^2}}$$
\end{theorem}

Towards our main goal of showing polynomial time convergence of the noisy dynamics (\ref{eq.gn}) (shown in Theorem
\ref{thm:convergencetime}), we need to show that the fitness increases within a few iterations of the
dynamic with high probability. It suffices to show that the average fitness under some transformation is a submartingale,
and then the result will follow using Azuma's inequality. 


\begin{lemma}\label{lem:submartingale} Let $\Phi^{t}$ be the random variable which corresponds to the average fitness at time $t$.
Assume that for the time interval $t=0,...,2T$ the trajectory $(\xx(t),\yy(t))$ has the same support. Let $m =
\max\{|\supp(\xx(t))|,$ $|\supp(\yy(t))|\}$, and 
the non-zero entries of $(\xx(t),\yy(t))$ be at least $\delta$. If $\frac{1}{(m+2)}(\frac{\gamma\delta m}{2}-2\alpha)^2 \geq \delta
\alpha^2$ then we have that $$E[\Phi^{2t+2} | \Phi^{2t},...,\Phi^{0}] \geq \Phi^{2t}+C\delta \alpha^2. $$ In other words, the sequence
$Z^{t}\equiv \Phi^{2t}-t\cdot C\delta \alpha^2$ for $t=1,...,T$ is a submartingale and also $|Z^{t+1}-Z^{t}| \leq W_{\max}-W_{\min}$.  
\end{lemma}

Using all the above analysis and Azuma's inequality (Theorem \ref{thm:azuma}), we establish our first main result on convergence
time of the noisy dynamics governed by (\ref{eq.gn}) for sexual evolution under natural selection and without mutation.

\begin{theorem} \label{thm:convergencetime}[Main 2] For all initial conditions $(\xx(0),\yy(0))\in \Delta$, the dynamics governed by
(\ref{eq.gn}) in an environment represented by fitness matrix $W$ 
reaches a pure fixed point with probability $1 - \epsilon$ 
after $O\left(\frac{(W_{\max})^4 n \ln (\frac{2n}{\epsilon})}{\delta^6\gamma^4}\right)$ iterations.  \end{theorem}
\begin{proof} 
It suffices to show that support size of the $\xx$ or $\yy$ reduces by one in a bounded number of iterations with at least $1 -
\frac{\epsilon}{2n}$ probability.  

Using Lemma \ref{lem:submartingale} we have that the random variable $\Phi^{2t} - t\cdot C\delta\alpha^2$ is a submartingale and since
$W_{\min}\leq \Phi^t \leq W_{\max}$ we use Azuma's inequality \ref{thm:azuma} and we get that  $$\Pr{\Phi^{2t} - t\cdot C\delta\alpha^2
\leq \Phi^0 - \lambda} \leq e^{-\frac{\lambda^2}{2t W^2_{max}}},$$ hence for $\lambda = \sqrt{2tW^2_{max}\ln (\frac{2n}{\epsilon})}$ we
get that the average fitness after $2t$ steps will be at least $\Phi^0-\sqrt{2tW^2_{max}\ln (\frac{2n}{\epsilon})}+ t\cdot C\delta
\alpha^2$ with probability at least $1 - \frac{\epsilon}{2n}$. By setting $t \geq \frac{8W^2_{\max}}{C^2\delta^2\alpha^4} \ln
\left(\frac{2n}{\epsilon}\right)$ we have that the average fitness at time $2t$ will be greater than $W_{\max}$ with probability $1 -
\frac{\epsilon}{2n}$, but since the potential is at most $W_{\max}$ for all vectors in the simplex, it follows that at some point the
frequency vector become {\em negligible}, { i.e.,} a coordinate of $\xx$ or $\yy$ became less than $\delta$. Hence, the probability
that the support size decreased during the process is at least $1 - \frac{\epsilon}{2n}$. 

By union bound (the initial support
size is at most $2n$) we conclude that dynamics (\ref{eq.gn}) reaches a pure
fixed point with probability $1 - \epsilon$ after $t$ iterations with $t = 2n\frac{8W^2_{\max}}{C^2\delta^2\alpha^4} \ln
\left(\frac{2n}{\epsilon}\right)$. Finally, for assumption $\frac{1}{(m+2)}(\frac{\gamma\delta m}{2}-2\alpha)^2 \geq
\delta \alpha^2$ used in Lemma \ref{lem:submartingale} to hold for $2 \leq m \leq n$, we set $\alpha$ 
to be such that $\alpha \leq \frac{\gamma\delta}{4}$ where we have $\frac{4(m-1)^2}{(m+2)} \geq1 > \delta$. Using such an $\alpha$ 
it follows that dynamics (\ref{eq.gn}) reaches a pure fixed point with probability $1 - \epsilon$ after $\frac{2^{18}}{9}\times \frac{
n W^4_{\max}}{\delta^6\gamma^4} \ln \left(\frac{2n}{\epsilon}\right)$ iterations.
\end{proof}

%% file: survival3.tex
\section{Changing Environment: Survival or Extinction?}\label{sec:survival}
In this section we analyze how evolutionary pressures under changing environment may lead to survival/extinction depending on the underlying mutation level. Motivated from Wolf \etal work \cite{wolf2005diversity}, we use Markov chain based model to capture the
 changing environment,
where every state captures a particular environment (see Section \ref{sec:envchange} for details).

\subsection{Extinction without mutation} We show that the population goes extinct with probability one, if the evolution is governed
by (\ref{eq.gn}), { i.e.,} natural selection {\em without} mutations under sexual reproduction. 
The proof of this result critically relies on polynomial-time convergence to monomorphic population shown in
Theorem \ref{thm:convergencetime} in case of fixed environment. 

As discussed in Section \ref{sec:envchange}, we have assume that the Markov chain is such that 
no individual can be fit to survive in all environments. Formally,
\begin{equation}\label{eq.env1}
\forall i,j,\ \ \ \prod_{e \in \CE} (W_{ij}^e)^{\pi_e} <1.
\end{equation}
Thus, if we can show convergence to monomorphic population under evolving environments as well then the extinction is guaranteed using
(\ref{eq.env1}) and the fact that population size $N^t$ gets multiplied by current average fitness (see (\ref{eq.fit})). 
However, showing convergence in stochastically changing environment is tricky because environment can change in any step with some
probability and then the argument described in the previous section breaks down. To circumvent this we will make use of Borel-Cantelli theorem where we say that {\em an event happens} if environment remains unchanged for a large but fixed number of steps.

\begin{theorem}\label{thm:borelcantelli} [Second Borel-Cantelli \cite{feller2008introduction}] Let $E_1,E_2,...$ be a sequence of
events. If the events $E_n$ are independent and the sum of the probabilities of the $E_n$ diverges to infinity, then the probability
that infinitely many of them occur is 1.
\end{theorem}


Using the above theorem with appropriate {\em event} definition, we prove the first part of the main result stated in {Theorem 1}. 

\begin{theorem}\label{thm:nomutationdie}[Main 1a] Regardless of the 
initial distributions $(\vec{x}(0),\vec{y}(0)) \in \Delta$, the population goes extinct with probability one 
under dynamics governed by (\ref{eq.gn}), capturing sexual evolution without mutation under natural selection. 
\end{theorem}
\begin{proof} Let $T^e$ be the number of iterations the dynamics (\ref{eq.gn}) need to reach a pure fixed point with probability
$\frac{1}{2}$. Theorem \ref{thm:convergencetime} implies $T^e = O\left ( \frac{ n W^{e^4}_{\max}}{\delta^6\gamma^{e^4}}\ln 4n\right )$. 
Let $T=\max_e T^e$. We consider the time intervals $1,...,T$, $T+1,...,2T$,... which are multiples of
$T$. The probability that Markov chain will remain at a specific environment $e$ in the time interval $kT+1,...,(k+1)T$ is $\rho_k =
(1-p)^T$. We define the sequence of events $E_1,E_2,...,$ where $E_i$ corresponds to the fact that the chain remains in the same
environment from time $(i-1) T+1,...,i T$. It is clear that $E_i$'s are independent and also $\sum_{i=1}^{\infty} \Pr{E_i} =
\sum_{i=1}^{\infty} \rho_i = \infty$. From Borel-Cantelli Theorem \ref{thm:borelcantelli} it follows that $E_i$'s happen
infinitely often with probability 1. When $E_i$ happens there is a time interval of length $T$ that the chain remains in the same
environment and therefore with probability $\frac{1}{2}$ the dynamics will reach a pure fixed point. After
$E_i$ happen for $k$ times, the probability to reach a pure fixed point is at least $1 - \frac{1}{2^k}$. Hence with probability one
(letting $k \to \infty$), the dynamics (\ref{eq.gn}) will eventually reach a pure fixed point. 

To finish the proof, let $T_{pure}$ be a random variable that captures the time when a pure fixed point, say $(i,j)$, is reached.
The population will have size at most $N^0 V^{T_{pure}}$ where $V = \max_e W^e_{\max}$. Under the assumption on the entries (see
inequality (\ref{eq.env1})) it follows that at any time $T'$ sufficiently large 
we get that the population at time $T'+T_{pure}$ will be roughly at most $$N^0 V^{T_{pure}} \prod_{e}(W^e_{ij})^{T'\pi_e} = N^0 V^{T_{pure}}
\left(\prod_{e}(W^e_{ij})^{\pi_e}\right )^{T'}.$$ By choosing $T' \geq \frac{\ln (N^0V^{T_{pure}})}{- \ln \left((W^e_{ij})^{\pi_e}\right )}$
(and also satisfying the constraint that is much greater than the mixing time) it follows that $N^{T'+T_{pure}} <1$ and hence the population dies.
So, the population goes extinct with probability one in the dynamics without mutation.  \end{proof}

\subsection{Survival with mutation}
In this section we consider evolutionary dynamics governed by (\ref{eq.fn}) capturing sexual evolution {\em with} mutation under
natural selection. Contrary to the case where there are no mutations we show that population survives
with positive  probability. Furthermore, this result turns out to be
robust in the sense that it holds even when every environment has some (few) very bad type alleles. Also, the result is independent of
the starting distribution of the population. 

The main intuition behind  proving this result is that, as for the mutation model in \cite{Hofbauer98}, every
allele is carried by at least $\tau$ fraction of the population in every generation. Therefore even if  a``good'' allele becomes ``bad''
as the environment changes, as far as the new environment has a few fit alleles, there will be some individuals carrying those who will then
procreate fast, spreading their alleles further and leading to overall survival. However, unlike in the no mutation case \cite{ITCS15MPP}, average
fitness is no more a potential function even for non-noisy dynamics, { i.e.,} it may decrease, and therefore showing such an
improvement is tricky.


First we show that if some small amount of time is spent in an environment then the frequencies of the bad alleles become small
and their effect is negligible, independent of the population distribution at the time when this environment was entered. 
Recall the assumption on good/bad type alleles (Section \ref{sec:envchange}). Formally, let $B^e_i$ be the set of bad type alleles
for $i=1,2$ in environment $e$,

\begin{equation}\label{ass:avfitness}
\begin{array}{c}
\forall i \in S_1\setminus B^e_1,\ \frac{\sum_j W^e_{ij}}{n} \ge 1+\beta,\ \ \ \mbox{and}\ \ \ \forall i \in S_1\setminus B^e_1, \forall
k \in B^e_1, W^e_{ij}\ge W^e_{kj},\ \forall j \\
\forall j \in S_2\setminus B^e_2,\ \frac{\sum_i W^e_{ij}}{n} \ge 1+\beta,\ \ \ \mbox{and}\ \ \ \forall j \in S_2\setminus B^e_2, \forall
k \in B^e_2, W^e_{ij}\ge W^e_{ik},\ \forall i \\
\end{array}
\end{equation}

\begin{lemma} \label{lem:badallele}Suppose that the environment $e$ is static for time at least $t \geq \frac{\ln (2n)}{n \tau}$.
For any $(\vec{x}(0),\vec{y}(0)) \in \Delta$, we have that $\sum_{i \in B^e_1} x_{i}(t) + \sum_{j \in B^e_2} y_{i}(t) \leq  \frac{2(|B^e_1|
+ |B^e_2|)}{n} = \frac{2|B^e|}{n}$ with $B^e = B^e_1\cup B^e_2$.
\end{lemma}

Using the fact that number of individuals with bad type alleles decreases very fast, established in Lemma \ref{lem:badallele}, we can prove that
within an environment  while there may be decrease in average fitness initially, this decrease is lower bounded. Moreover, it will later increase fast
enough so that the initial decrease is compensated.

\begin{lemma}\label{lem:phasetransition} Suppose that the environment $e$  is static for time $t$ and also $\tau \leq \frac{\beta}{16n}$, $
|B^e| \ll n \beta$ then  there exists a threshold time $T_{thr}$ such that for any given initial distributions of the alleles
$(\vec{x}(0),\vec{y}(0)) \in \Delta$, if $t < T_{thr}$ then the population size will experience a loss factor of at most $\frac{1}{d}$,
otherwise it will experience a gain factor of at least $d$ for some $d>1$, where $T_{thr}=\frac{6  \ln (2n)}{n\tau\beta W_{\min}}$ and
$W_{\min} = \min_e W_{\min}^e$. 
\end{lemma}

To show the second part of {Theorem 1} (main result),  
we will couple the random variable corresponding to the number of individuals
at every iteration with a biased random walk on the real line. This can be done since in Lemma \ref{lem:phasetransition} we established that the
decrease and increase in average fitness is upper and lower bounded, respectively. We will apply the following well-known lemma about the
biased random walks.

\begin{lemma}\label{lem:biased} (Biased random walk) Assume we perform a random walk on the real line, starting from point $k \in
\mathbb{N}$ and going right (+1) with probability $q>\frac{1}{2}$ and left (-1) with probability $1-q$. The probability that we will
eventually reach 0 is $\left(\frac{1-q}{q}\right)^k$.  \end{lemma}

Using Lemma \ref{lem:phasetransition} together with the biased random walk Lemma \ref{lem:biased}, we show our next main result on
survival of population under mutation in the following theorem.

\begin{theorem}\label{thm:mutationsurvive}[Main 1b] If $p < \frac{1}{2T_{thr}}$ where $T_{thr}= \frac{6\ln (2n)}{n\tau \beta
W_{\min}}$ then the probability of survival is at least $1-\left(\frac{pT_{thr}}{1-pT_{thr}}\right)^{c\ln N^0 }$ for some $c$
independent of $N^0$, $c = \left(\frac{n\tau W_{\min}}{\ln (2n)}\right)$.  \end{theorem}

\begin{proof} The probability that the chain remains at a specific environment for least $T_{thr}$ iterations is $(1-p)^{T_{thr}} > 1 -
pT_{thr}$ (from the moment it enters the environment until it departs) and hence the probability that the chain stays at an environment
for time less that $T_{thr}$ is at most $pT_{thr}$. Let $N^t = N^0\prod_{j=1}^t \vec{x}(j)^TW^{e(j)}\vec{y}(j)$ (see (\ref{eq.fit}) where
here $e(j)$ corresponds to the environment at time $j$) the number of individuals at time $t$ and $Z^i$ be the position of the biased random walk at time $i$ as defined in Lemma \ref{lem:biased} with $q =1-pT_{thr}$ and assume that $Z^0 = \lfloor \log_d N^0 \rfloor$ ($d$ is from lemma \ref{lem:phasetransition}). Let $t_1,t_2,...$ be the sequence of times where there is a change of environment (with $t_0=0$) and consider the trivial coupling where when the chain changes environment then a move is made on the real line. If the chain remained in the environment for time less than $T_{thr}$ then the walk goes left, otherwise it goes right. It is clear by Lemma \ref{lem:phasetransition} that random variable $\log_d N^{t_i}$ dominates $Z^i$. Hence, the probability that the population survives is at least the probability that $Z^i$ never reaches zero ($Z^i>0$ for all $i \in \mathbb{N}$). By Lemma \ref{lem:biased} this is at most $(\frac{pT_{thr}}{1-pT_{thr}})^{\lfloor \log_d N^0 \rfloor}$ and thus the probability of survival is at least $1-\left(\frac{pT_{thr}}{1-pT_{thr}}\right)^{c\ln N^0 }$ where $c = \left(\frac{n\tau W_{\min}}{\ln (2n)}\right)$ depends on
$n,\tau$ and fitness matrices $W^e$ (the minimum $W_{\min} = \min_e W^e_{\min}$, and also from Lemma \ref{lem:phasetransition} we  have that $\ln d \approx \frac{\ln 2n}{W_{\min}n\tau}$).  \end{proof}

%% file: mutation1.tex
\section{Convergence of Discrete Replicator Dynamics with Mutation in Fixed Environments}
\label{sec:mutation}

In this section we extend the  convergence result of Mehta \etal \cite{ITCS15MPP} for dynamics (\ref{eq.g}) in static environment 
to dynamics governed by (\ref{eq.f}) where mutations are also present. The former result critically hinges on the fact that mean fitness strictly increases unless
the system is at a fixed-point, and thereby acts as a potential function. 
Despite the fact that this is no longer the case when mutations are introduced, we manage to show that 
the system still converges and follows an intuitively clear behavior. Namely, in every step of the dynamic, either the average fitness
$\vec{x}^TW\vec{y}$ or the product of the proportions of all different alleles $\prod_{i}x_{i}\prod_{i}y_{i}$ (or both) will increase.
This latter quantity is, in some sense, a measure of how mixed/diverse the population is.
To argue this we apply the following inequality due to Baum and Eagon:

\begin{theorem}[\textbf{Baum and Eagon Inequality}~\cite{BE66}]
  \label{thm:inequality} Let $P(\vec{x})=P(\{x_{ij}\})$ be a polynomial with nonnegative
coefficients homogeneous of degree $d$ in its variables $\{x_{ij}\}$. Let $\vec{x}=\{x_{ij}\}$
be any point of the domain $D: x_{ij}\geq 0, \sum_{j=1}^{q_i} x_{ij}= 1, i = l,\ldots, p,
j=l,\ldots,q_i$. For $\vec{x}= \{{x_{ij}}\} \in D$, let $\Xi(\vec{x}) = \Xi\{x_{ij}\}$ denote the point of $D$
whose $i, j$ coordinate is
  $$
    \Xi(\vec{x})_{ij} = \left (x_{ij} \frac{\partial P}{\partial x_{ij}} \bigg|_{(\vec{x})}\right) \cdot  \left( \sum_{j=1}^{q_i}x_{ij} \frac{\partial P}{\partial x_{ij}} \bigg|_{(\vec{x})}\right)^{-1}.
$$
  Then $P(\Xi(\vec{x}))>P(\vec{x})$ unless $\Xi(\vec{x})=\vec{x}$.
\end{theorem}
\noindent

We will establish a potential function $P$ that for the dynamics governed by (\ref{eq.f}), capturing sexual evolution with mutation.
 This will imply convergence for
the dynamics. Note that feasible values of $\tau$ are in $[0, \frac{1}{n}]$, since $\tau$ represents fraction of allele $i$ mutating to
allele $i'$ of the same gene implying $n * \tau\le 1$. 

\begin{theorem}\label{thm:potential}[Main 3]
Given a static environment $W$, dynamics governed by (\ref{eq.f}) with mutation parameter $\tau \leq \frac{1}{n}$ has a potential
function $P(\xx,\yy) = (\xx^TW\yy)^{1 - n\tau}\prod_{i}x_{i}^{\tau}\prod_{i}y_{i}^{\tau}$ that strictly increases unless an equilibrium
(fixed-point) is reached. 
Thus, the system converges
to equilibria, in the limit. Equilibria are exactly the set of points $(\vec{p}^*,\vec{q}^*)$ that satisfy for all $i,i'\in S_1$, $j,j' \in S_2$:
$$ \frac{(W\vec{q}^*)_{i}}{1 - \frac{\tau}{p^*_{i}}} = \frac{(W\vec{q}^*)_{i'}}{1 - \frac{\tau}{p^*_{i'}}} = \frac{\vec{p}^{*T}W\vec{q}^{*}}{1-n\tau} = \frac{(W^T\vec{p}^*)_{j}}{1 - \frac{\tau}{q^*_{j}}} = \frac{(W^T\vec{p}^*)_{j'}}{1 - \frac{\tau}{q^*_{j'}}}$$.
\end{theorem}

\noindent
As a consequence of the above theorem we get the following:

\begin{corollary}
Along every nontrivial trajectory of dynamics governed by (\ref{eq.f}) 
at least one of average population fitness $\xx^TW\yy$ or
product of allele frequencies $\prod_{i}x_{i}\prod_{i}y_{i}$ strictly increases at each step.  \end{corollary}

%% file: tight.tex
\section{Discussion on the Assumptions and Examples}\label{sec:full}

In this section, we will discuss why our assumptions are necessary and their significance.

\subsection{On the parameters $\gamma,\delta,\beta, \tau$}\label{ass:time}

The effective range of 
 $\delta$ is $o\left (\frac{1}{n}\right )$, where $||\vec{\delta}||_{\infty} = \delta$, whereas for  $\gamma$ is $O\left (\frac{1}{n^2}\right )$. For example, if we consider the entries of fitness matrices $W^e$ to be uniform from interval $(1-\sigma,1+\sigma)$ for some positive $\sigma>0$ then $\gamma$ is of $\Theta(\frac{1}{n^2})$ order. If the entries of the matrix are constants (in weak selection scenario they lie in the interval $(1-\sigma,1+\sigma)$) then the convergence time of dynamics \ref{eq.gn} is polynomial w.r.t $n$ (size of fitness matrices $W^e$ is $n \times n$). We note that the main result of \cite{ITCS15MPP} for dynamics (\ref{eq.g}) has been derived under the assumption that the entries of the fitness matrix are all distinct. It is proven that this assumption is necessary by giving  examples where the dynamic doesn't converge to pure fixed points if the fitness matrix has some entries that are equal (the trivial example is when $W$ has all entries equal, then every frequency vector in $\Delta$ is a fixed point). This is an indication that $\gamma$ is needed to analyse the running time and is not artificial.
 The noise vector $\vec{\delta}$ has coordinates $\pm \delta$, so it is uniformly chosen from hypercube, but there is no dependence on the current frequency vector ($\delta$ is independent of current $(\vec{x},\vec{y})$). Finally, $\beta$ should be thought of  as a small constant number (like in weak selection) independent of $n$, and $\tau$ to be $O(\frac{1}{n})$ ($1-n\tau \geq 0$ must hold so that the dynamics with mutation are meaningful and from Lemma \ref{lem:phasetransition}, it must hold that $\tau \leq \frac{\beta}{16n}$).

\subsection{On the environments}\label{ass:environment}
We analyze a finite population model where $N^t$ is the population size at time $t$. It is natural to define survival if $N^t \geq 1$
for all $t \in \mathbb{N}$ (number of people is at least 1 at all times) and extinction if $N^t<1$ for some $t$ (if the number of people is
less than one at some point then the population goes extinct). As described in preliminaries, $N^t = N^{t-1} \cdot \Phi^t$ where
$\Phi^t = x(t)^TW^{e(t)}y(t)$ is the average fitness at time $t$ and $W^{e(t)}$ is the fitness matrix of environment $e(t)$ (at time $t$). 

Fix a fitness matrix $W$ (i.e., fix an environment). If $W_{ij}> 1+\epsilon$ for all $(i,j)$ then $\vec{x}^TW\vec{y} \geq 1$ for all $(\vec{x},\vec{y})\in \Delta$  and thus
the number of individuals is increasing along the generations by a factor of $1+\epsilon$ (the population survives). On the other
hand, if $W_{ij}<1-\epsilon$ for all $(i,j)$ then $\vec{x}^TW\vec{y} < 1-\epsilon$ for all $(\vec{x},\vec{y})\in \Delta$, so it is clear
that the number of individuals is decreasing with a factor of $1-\epsilon$ (thus population goes extinct). So either extreme makes the
problem irrelevant. 

Finally, it is natural to assume that complete diversity should favor survival, {i.e.,}
if the population is uniform along the alleles/types then the population size must not decrease in the next generation.
Therefore, we assume that the average fitness under
uniform frequencies is $\geq 1+\beta$ (for all but few number of bad alleles that can be seen as deleterious). The alleles that are good should dominate entry-wise the bad alleles. Example Figure \ref{fig2} in the appendix $C$ shows that this assumption is necessary. In Figure \ref{fig2}, $\tau =
0.03$ and $W^e = \left(\begin{array}{cc} 0.99 & 0.37\\ 0.56 & 2.09\end{array}\right).$ If we start from any vector $(\vec{x},\vec{y})$
in the shaded area, the dynamics converges to the stable fixed point $B$. The average fitness $\vec{x}^TW\vec{y}$ at $B$ is less than the maximum at the corner which is $W^{e}_{1,1} = 0.99<1$. So if the size of population is $Q$ when entering $e$, after $t$ generations on the environment $e$, the
population size will be at most $Q \cdot 0.99^t$ (which decreases exponentially). In that case Theorem \ref{thm:mutationsurvive} doesn't hold,
even though $\frac{0.99+0.37+0.56+2.09}{4}=1.0025>1$ and $\beta = 0.0025$ (qualitatively we would have the same picture for any $\tau
\in [0,0.03]$ and $W^e$).

The assumption defined in (\ref{eq.env1}) is necessary as well for the following reason: Assume there is a combination of alleles $(i,j)$
so that $\prod_{e} (W^{e}_{ij})^{\pi_e}\geq 1$ (*). In that case we can have one of the environments so that $x_i=1,y_j=1$ is a stable
fixed point and hence there are initial frequencies so that the dynamics (\ref{eq.gn}) converge to it. After that, it is easy to argue
that this monomorphic population survives on average because of (*), so the probability of survival in that case is non zero.
\subsection{Explanation of Figure \ref{fig1}}
Figure \ref{fig1} on the title page shows the adjacency graph of a Markov chain. There are 3 environments with fitness
matrices, say $W^{e_1}, W^{e_2}, W^{e_3}$, and the entries of every matrix are distinct. Take $p_{ii} = 1- p$ and $p_{ij} =
\frac{p}{2}$ so that the stationary distribution is $(1/3,1/3,1/3)$. Observe that $W^{e_1}_{1,1} \cdot W^{e_2}_{1,1} \cdot W^{e_3}_{1,1} = 1.12\cdot
1.02\cdot 0.87<0.994<1$. The same is true for entries (1,2),(2,1),(2,2). So the assumption defined in (\ref{eq.env1}) is satisfied.

Moreover, observe that if we choose $\beta = 0.005$ and hence $\tau = \frac{0.005}{32}$ it follows that the assumptions
defined in (\ref{ass:avfitness}) are satisfied (also the bad alleles are dominated entry-wise by the good alleles). Hence, in case of no mutation,
from theorem \ref{thm:nomutationdie} the population dies out with probability 1 for all initial population sizes $N^0$ and all initial
frequency vectors in $\Delta$. In case of mutation, and for sufficiently large initial population size $N^0$, for all initial frequency
vectors in $\Delta$ the probability of survival is positive (Theorem \ref{thm:mutationsurvive}).

%% file: appendix.tex

\section*{Appendix}
\section{Terms Used in Biology}\label{asec.bioTerms}
We provide brief non-technical definitions of a few biological terms useful for this paper.

\noindent{\bf Gene.}
A unit that determines some characteristic of the organism, and passes traits to offsprings.
All organisms have genes corresponding to various biological traits, some of which are instantly visible, such as eye color or number
of limbs, and some of which are not, such as blood type.
\medskip

\noindent{\bf Allele.}
Allele is one of a number of alternative forms of the same gene, found at the same place on a chromosome,
Different alleles can result in different observable traits, such as different pigmentation.
\medskip

\noindent{\bf Genotype.}
The genetic constitution of an individual organism.
\medskip

\noindent{\bf Phenotype.}
The set of observable characteristics of an individual resulting from the interaction of its genotype with the environment.
\medskip

\noindent{\bf Diploid.}
Diploid means having two copies of each chromosome. Almost all of the cells in the human body are diploid.
\medskip

\noindent{\bf Haploid.}
A cell or nucleus having a single set of unpaired chromosomes.
Our sex cells (sperm and eggs) are haploid cells that are produced by meiosis. When sex cells unite during fertilization, the haploid
cells become a diploid cell.

\input{proofs}

\subsection{Calculations for mutation}\label{app:calc}
Let $(\hat{\vec{x}},\hat{\vec{y}}) = g(\vec{x},\vec{y})$. If in every generation allele $i \in S_1$ mutates to allele $k\in S_1$ with probability $\mu_{ik}$, where $\sum_k \mu_{ik}=1,\ \forall i$, then
the final proportion (after reproduction, mutation) of allele $i \in S_1$ in the population will be

\[
x'_i = \sum_{k\in S_1} \mu_{ki}\hat{x}_k
\]
Similarly, if $j\in S_2$ mutates to $k \in S_2$ with probability $\delta_{jk}$, then proportion of allele $j \in S_2$ will be

\[
y'_j = \sum_{k\in S_2} \delta_{ki}\hat{y}_k
\]

If mutation happens after every selection (mating), then we get the following dynamics with update rule $f':\Delta\rightarrow \Delta$ governing the evolution (update rule contains selection+mutation).

\begin{equation}\label{eq.f}
\mbox{Let $(\vec{x}',\vec{y}')=f'(\vec{x},\vec{y})$, then }\begin{array}{lcl}
x'_i & = & \sum_{k \in S_1} \mu_{ki} x_k \frac{(W\vec{y})_k}{\vec{x}^TW\vec{y}},\ \forall i\in S_1 \\
y'_j & = & \sum_{k \in S_2} \delta_{kj} y_k \frac{(\vec{x}^TW)_k}{\vec{x}^TW\vec{y}},\ \forall j\le S_2 \\
\end{array}
\end{equation}
Suppose $\forall k,\ \forall i\neq k$ and $\forall j \neq k$, we have $\mu_{ik} = \delta_{jk} = \tau$, where $\tau \leq \frac{1}{n}$.
Since $\sum_{k} \mu_{ik}=\sum_{k}\delta_{jk}=1$, we have $\mu_{ii}=\delta_{jj}=1-(n-1)\tau=1+\tau-n\tau$.
Hence
\begin{align*}
x'_i  &=  \sum_{k \in S_1} \mu_{ki} x_k \frac{(W\vec{y})_k}{\vec{x}^TW\vec{y}}
\\&= (1+\tau-n\tau)x_i\frac{(W\vec{y})_i}{\vec{x}^TW\vec{y}} + \tau \sum_{k\neq i}x_k\frac{(W\vec{y})_k}{\vec{x}^TW\vec{y}}
\\&=(1-n\tau) x_i\frac{(W\vec{y})_i}{\vec{x}^TW\vec{y}} +\tau \sum_{k}x_k\frac{(W\vec{y})_k}{\vec{x}^TW\vec{y}}
\\& = (1-n\tau) x_i\frac{(W\vec{y})_i}{\vec{x}^TW\vec{y}} +\tau.
\end{align*}
The same is true for vector $\vec{y}'$. The dynamics of (\ref{eq.f})  where $\mu_{ik} = \delta_{ik} = \tau$ for all $k\neq i$ simplifies to the equations \ref{eq.fs} as appear in the preliminaries.

%% file: proofs.tex
\section{Missing proofs}
\subsection{Proof of Lemma \ref{lem:inequality}}
\begin{proof} From the definition of $g$ (equation \ref{eq.g}) we get,
\begin{align*}
2\left(\hat{\vec{x}}^TW\hat{\vec{y}} \right) \left ( \vec{x}^TW\vec{y}\right)^2 &= 2 \sum _{ij} W_{ij}\hat{x}_i\hat{y}_j\left ( \vec{x}^TW\vec{y}\right)^2
\\&= 2\sum _{ij} W_{ij}x_i y_j\frac{(W\vec{y})_i}{\vec{x}^TW\vec{y}}\frac{(W^T\vec{x})_j}{\vec{x}^TW\vec{y}} \left ( \vec{x}^TW\vec{y}\right)^2
\\&= 2\sum_{i,j}W_{ij}x_{i}y_{j}(W\vec{y})_i (W^T\vec{x})_j \\&=
\sum_{i,j,k}W_{ij}W_{ik}x_{i}y_{j}y_{k} (W^T\vec{x})_j+\sum_{i,j,k}W_{ij}W^T_{jk}x_{i}x_{k}y_{j}(W\vec{y})_i\\&=
\sum_{i,j,k}W_{ij}W_{ik}x_{i}y_{j}y_{k} \frac{1}{2}((W^T\vec{x})_j+(W^T\vec{x})_k)+\sum_{i,j,k}W_{ij}W_{kj}x_{i}x_{k}y_{j}\frac{1}{2}((W\vec{y})_i+(W\vec{y})_k)\\&\geq
\sum_{i,j,k}W_{ij}W_{ik}x_{i}y_{j}y_{k} \sqrt{(W^T\vec{x})_j(W^T\vec{x})_k}+\sum_{i,j,k}W_{ij}W_{kj}x_{i}x_{k}y_{j}\sqrt{(W\vec{y})_i(W\vec{y})_k}\\&=
\sum_{i}x_i\left(\sum_{j}y_jW_{ij}\sqrt{(W^T\vec{x})_j}\right)^2 + \sum_{j}y_j\left(\sum_{i}x_iW_{ij}\sqrt{(W\vec{y})_i}\right)^2\\&\geq
\left(\sum_{i,j}x_iy_jW_{ij}\sqrt{(W^T\vec{x})_j}\right)^2 + \left(\sum_{j,i}y_jx_iW_{ij}\sqrt{(W\vec{y})_i}\right)^2 \textrm{ using convexity of }f(z)=z^2\\&=
\left(\sum_{j}y_j (W^T\vec{x})_j^{3/2}\right)^2 + \left(\sum_{i}x_i (W\vec{y})_i^{3/2}\right)^2. \textrm{ \;\;\;\;\;\;\;\;\;\;\;\;\;\;\;\;\;\;\;\;\;\;\;\;\;\;\;\;\;\;\;\;\;\;\;\;\;\;\;\;\;\;\;\;\;\;\;\;\;\;\;\;\;\;\;\;(0)}
\end{align*}

Let $\xi$ be a random variable that takes value $(W\vec{y})_i$ with probability $x_i$. Then $\mathbb{E}[\xi] = \vec{x}^TW\vec{y}$,
$\mathbb{V}[\xi] = \sum_{i}x_{i}((W\vec{y})_i-\vec{x}^TW\vec{y})^2$ and $\xi$ takes values in the interval $[0,\mu]$ with $\mu =
\max_{ij}W_{ij}$. Consider the function $f(z) = z^{3/2}$ on the interval $[0,\mu]$ and observe that $f''(z) \geq
\frac{3}{4}\frac{1}{\sqrt{\mu}}$ on $[0,\mu]$ since $\mu \geq \vec{p}^TW\vec{q} \geq 0$ for all $(\vec{p},\vec{q}) \in \Delta$. Observe
also that $f(\mathbb{E}[\xi]) = (\vec{x}^TW\vec{y})^{3/2}$ and $\mathbb{E}[f(\xi)] = \sum_{i}x_i (W\vec{y})_i^{3/2}$.  \bigskip

\begin{claim}
$\mathbb{E}[f(\xi)] \geq f(\mathbb{E}[\xi]) + \frac{A}{2}\mathbb{V}[\xi]$, where $A = \frac{3}{4\sqrt{\mu}}$.
\end{claim}
\begin{proof}
By Taylor expansion we get that (we expand w.r.t the expectation of $\xi$, namely $\mathbb{E}[\xi]$)$$f(z) \geq f(\mathbb{E}[ \xi ])+ f'(\mathbb{E} [\xi])(z - \mathbb{E}[ \xi])+\frac{A}{2}(z - \mathbb{E}[ \xi])^2$$ and hence we have that:
\begin{align*}
f(z) &\geq f(\mathbb{E}[ \xi ])+ f'(\mathbb{E}[ \xi])(z - \mathbb{E}[ \xi])+\frac{A}{2}(z - \mathbb{E}[ \xi])^2\overbrace{\Rightarrow}^{\textrm{taking expectation}}\\
\mathbb{E} [f(\xi)] &\geq \mathbb{E}[f(\mathbb{E}[ \xi] )]+ f'(\mathbb{E} [\xi])(\mathbb{E}[\xi] - \mathbb{E}[ \xi])+\frac{A}{2}\mathbb{V}[\xi]\\&= f(\mathbb{E} [\xi] )+\frac{A}{2}\mathbb{V}[\xi].
\end{align*}
\end{proof}

Using the above claim 
it follows that:
$$\sum_{i}x_{i}(W\vec{y})_i^{3/2} \geq (\vec{x}^TW\vec{y})^{3/2} + \frac{3}{8\sqrt{\mu}}\sum_{i}x_{i}((W\vec{y})_i-\vec{x}^TW\vec{y})^2.$$ Squaring both sides and omitting one square from the r.h.s we get
\begin{equation}
\left(\sum_{i}x_{i}(W\vec{y})_i^{3/2}\right)^2 \geq (\vec{x}^TW\vec{y})^{3} + \frac{3}{4\sqrt{\mu}}(\vec{x}^TW\vec{y})^{3/2}\sum_{i}x_{i}((W\vec{y})_i-\vec{x}^TW\vec{y})^2.
\end{equation}
We do the same by setting $\xi$ to be $(W^T\vec{x})_i$ with probability $y_i$ and using similar argument we get
\begin{equation}
\left(\sum_{i}y_{i}(W^T\vec{x})_i^{3/2}\right)^2 \geq (\vec{x}^TW\vec{y})^{3} + \frac{3}{4\sqrt{\mu}}(\vec{x}^TW\vec{y})^{3/2}\sum_{i}y_{i}((W^T\vec{x})_i-\vec{x}^TW\vec{y})^2.
\end{equation}
Therefore it follows that
\begin{align*}
2(\hat{\vec{x}}^TW\hat{\vec{y}}) (\vec{x}^TW\vec{y})^2 &\geq \left(\sum_{j}y_j (W^T\vec{x})_j^{3/2}\right)^2 + \left(\sum_{i}x_i (W\vec{y})_i^{3/2}\right)^2\textrm{ by inequality } (0)\\&
\overbrace{\geq}^{\textrm{(9)+(10)}} 2(\vec{x}^TW\vec{y})^3 + \frac{3}{4\sqrt{\mu}}(\vec{x}^TW\vec{y})^{3/2}\left(\sum_{i}x_{i} \left((W\vec{y})_{i} - \vec{x}^TW\vec{y}\right)^2+\sum_{i}y_{i} \left((W^T\vec{x})_{i} - \vec{x}^TW\vec{y}\right)^2\right )
\end{align*}
Finally we devide both sides by $2(\vec{x}^TW\vec{y})^2$ and we get that
\begin{align*}
(\hat{\vec{x}}^TW\hat{\vec{y}}) &\geq (\vec{x}^TW\vec{y}) + \frac{3}{8\sqrt{\mu(\vec{x}^TW\vec{y})}}\left(\sum_{i}x_{i} \left((W\vec{y})_{i} - \vec{x}^TW\vec{y}\right)^2+\sum_{i}y_{i} \left((W^T\vec{x})_{i} - \vec{x}^TW\vec{y}\right)^2\right )
\\&\geq (\vec{x}^TW\vec{y}) + \frac{3}{8\mu}\left(\sum_{i}x_{i} \left((W\vec{y})_{i} - \vec{x}^TW\vec{y}\right)^2+\sum_{i}y_{i} \left((W^T\vec{x})_{i} - \vec{x}^TW\vec{y}\right)^2\right )
\end{align*}

with $\frac{3}{8\sqrt{\mu\Phi(\vec{x},\vec{y})}} \geq \frac{3}{8\mu}$ since $\mu \geq \vec{x}^TW\vec{y}$. This inequality and the proof techniques can be seen as a generalization of an inequality and proof techniques in \cite{newlyo}.
\end{proof}
\subsection{Proof of Lemma \ref{lem:zeroinexpectation}}
\begin{proof} Vectors $(\vec{\delta}_{\vec{x}},\vec{\delta}_{\vec{y}}),
(-\vec{\delta}_{\vec{x}},\vec{\delta}_{\vec{y}}),(\vec{\delta}_{\vec{x}},-\vec{\delta}_{\vec{y}}),
(-\vec{\delta}_{\vec{x}},-\vec{\delta}_{\vec{y}})$ appear with the same probability, and observe that
$$(\vec{x+\vec{\delta}_{\vec{x}}})^TW(\vec{y}+\vec{\delta}_{\vec{y}})+(\vec{x-\vec{\delta}_{\vec{x}}})^TW(\vec{y}+\vec{\delta}_{\vec{y}})+(\vec{x+\vec{\delta}_{\vec{x}}})^TW(\vec{y}-\vec{\delta}_{\vec{y}})+(\vec{x-\vec{\delta}_{\vec{x}}})^TW(\vec{y}-\vec{\delta}_{\vec{y}})
= 4\vec{x}^TW\vec{y},$$ and the claim follows.
\end{proof}
\subsection{Proof of Lemma \ref{lem:goodjump}}
\begin{proof} Assume w.l.o.g that we have $W_{i1} \geq W_{i2} \geq ...$ (otherwise we permute them so that are in decreasing order). Consider the case where the signs are revealed one at a time, in the order of indices of the sorted row. The probability that $+$ signs dominate $-$ signs through the process is $\frac{1}{m/2+1}$ (ballot theorem/Catalan numbers) (see \ref{lem:ballot}). It is clear that when the $+$ signs dominate the $-$ signs then
\begin{align*}
(W\vec{\delta}_{\vec{y}})_i &= \sum_j^{m}W_{ij}\delta_j \ \ \geq\ \  \sum_{j=1}^{m/2} (W_{i(2j-1)}-W_{i(2j)})\delta \ \ \geq \ \
\gamma\delta\frac{m}{2} \end{align*}
\end{proof}

\subsection{Proof of Lemma \ref{lem:submartingale}}
\begin{proof}
First of all, since the average fitness is increasing in every generation (before adding noise) and by Lemma \ref{lem:zeroinexpectation} we get that for all $t \in \{0,...,2T\}$ $$\mathbb{E}[\Phi^{t+1} | \Phi^{t}] \geq \Phi^{t}$$ namely the average fitness is a submartingale (0).\\\\

Let $(\vec{x}^{t},\vec{y}^{t}) \defeq (\vec{x}(t),\vec{y}(t))$ be the frequency vector at time $t$ which has average fitness $\Phi^{t} \equiv \Phi(\vec{x}^{t},\vec{y}^{t}) = {\vec{x}^{t}}^TW\vec{y}^{t}$ (abusing notation we use $\Phi(\vec{x},\vec{y})$ for function $\vec{x}^TW\vec{y}$ and $\Phi^{t}$ for the value of average fitness at time $t$), also we denote $(\hat{\vec{x}}^{t},\hat{\vec{y}}^t) = g(\vec{x}^t,\vec{y}^t)$ and recall that $(\vec{x}^{t+1},\vec{y}^{t+1}) = (\hat{\vec{x}}^{t}+\vec{\delta}^{t}_{\vec{x}},\hat{\vec{y}}^{t}+\vec{\delta}^{t}_{\vec{y}})$. Assume that in the next generation $(\hat{\vec{x}}^{2t},\hat{\vec{y}}^{2t}) = g(\vec{x}^{2t},\vec{y}^{2t})$ the average fitness before the noise, namely $\hat{\vec{x}}^{2t\;T}W\hat{\vec{y}}^{2t}$ will be at least $\Phi^{2t}+C \delta \alpha^2$. Hence by Lemma \ref{lem:zeroinexpectation} we get that $\mathbb{E}[\Phi^{2t+1} | \Phi^{2t}] = \hat{\vec{x}}^{2t\;T}W\hat{\vec{y}}^{2t} \geq \Phi^{2t}+ C\delta\alpha^2$ (1). Therefore we have that
\begin{align*}
\mathbb{E}[\Phi^{2t+2}| \Phi^{2t}] & = \mathbb{E}_{\vec{\delta}^{2t+1},\vec{\delta}^{2t}}[ (\hat{\vec{x}}^{2t+1}+\vec{\delta}^{2t+1}_{\vec{x}})^TW(\hat{\vec{y}}^{2t+1}+\vec{\delta}^{2t+1}_{\vec{y}}) |\Phi^{2t}]
\\&=\mathbb{E}_{\vec{\delta}^{2t}}[ \left(\hat{\vec{x}}^{2t+1}\right)^TW\hat{\vec{y}}^{2t+1} |\Phi^{2t}]
\\&\geq \mathbb{E}_{\vec{\delta}^{2t}}[ \left(\vec{x}^{2t+1}\right)^TW\vec{y}^{2t+1} |\Phi^{2t}]
\\&= \mathbb{E}[\Phi^{2t+1} | \Phi^{2t}]
\\& \geq \Phi^{2t}+ C\delta\alpha^2
\end{align*}
where the second inequality is claim (1) and the first inequality comes from inequality \ref{lem:inequality} (since the r.h.s of inequality \ref{lem:inequality} is non-negative). The first, third equality comes from model definition and second equality comes from Lemma \ref{lem:zeroinexpectation}.\\\\
Assume now that in the next generation $(\hat{\vec{x}}^{2t},\hat{\vec{y}}^{2t}) = g(\vec{x}^{2t},\vec{y}^{2t})$ the average fitness before the noise, namely $\hat{\vec{x}}^{2t\;T}W\hat{\vec{y}}^{2t}$ will be less than $\Phi^{2t}+C \delta \alpha^2$. This means that the vector $(\vec{x}^{2t},\vec{y}^{2t})$ is $\alpha$-close by corollary \ref{cor:notaclose}, so after adding the noise by the definition of $\alpha$-close we get that $\hat{\vec{x}}^{2t\;T}W\hat{\vec{y}}^{2t}+\alpha \geq \Phi^{2t+1}\geq \hat{\vec{x}}^{2t\;T}W\hat{\vec{y}}^{2t}-\alpha$ (2). From Lemma \ref{lem:goodjump} we will have with probability at least $\frac{1}{2}\frac{1}{m/2+1}$ that $(W\vec{y}^{2t+1})_i \geq (W\hat{\vec{y}}^{2t})_i +\frac{\gamma\delta m}{2}$ for all $i$ in the support of vector $\vec{x}^t$ (we multiplied the probability by $\frac{1}{2}$ since you perturb $\vec{y}$ with probability half) (3). The same argument works if we purturb $\vec{x}$, so w.l.o.g we work with purturbed vector $\vec{y}$ which has support of size at least 2. Essentially by inequality \ref{lem:inequality} we get the following inequalities:
\begin{align*}
\mathbb{E}[\Phi^{2t+2}| \Phi^{2t}] & = \mathbb{E}_{\vec{\delta}^{2t+1},\vec{\delta}^{2t}}[ (\hat{\vec{x}}^{2t+1}+\vec{\delta}^{2t+1}_{\vec{x}})^TW(\hat{\vec{y}}^{2t+1}+\vec{\delta}^{2t+1}_{\vec{y}}) |\Phi^{2t}]
\\&=\mathbb{E}_{\vec{\delta}^{2t}}[ \left(\hat{\vec{x}}^{2t+1}\right)^TW\hat{\vec{y}}^{2t+1} |\Phi^{2t}]
\\&\overbrace{\geq}^{\ref{lem:inequality}}  \mathbb{E}_{\vec{\delta}^{2t}}[\left(\vec{x}^{2t+1}\right)^TW\vec{y}^{2t+1}|\Phi^{2t}]
+C\cdot \mathbb{E}_{\vec{\delta}^{2t}}\left[ \sum_i x_i^{2t+1}\cdot \left((W\vec{y}^{2t+1})_i -   \left(\vec{x}^{2t+1}\right)^TW\vec{y}^{2t+1}  \right)^2\bigg\lvert\Phi^{2t}\right]
\\&\geq   \Phi^{2t} +\frac{C}{m+2} \left(\frac{\gamma\delta m}{2} -  2\alpha\right)^2
\\&\geq \Phi^{2t} + C\delta \alpha^2
\end{align*}
where last inequality comes from the assumption and the second comes from claim (0), (2), (3). Hence by induction we get that $$\mathbb{E}[\Phi^{2t+2} - (t+1) \cdot C\delta\alpha^2 | \Phi^{2t}] \geq \Phi^{2t}-t \cdot C\delta\alpha^2.$$ It is easy to see that $W_{\max} \geq \Phi^t \geq W_{\min}$ for all $t$.
\end{proof}
\subsection{Proof of Lemma \ref{lem:badallele}}
\begin{proof} Consider one step of the dynamics that starts at $(\xx,\yy)$ and has frequency vector $(\tilde{\vec{x}},\tilde{\vec{y}})$ in the next step before adding the noise. Let $i*$ be the bad allele that has the greatest fitness
at it, namely $(W^e \vec{y})_{i*} \geq (W^e \vec{y})_i$ for all $i \in B^e_1$. It holds that
\begin{align*}
\sum_{i\in B^e_1} \tilde{x}_{i} &= (1-n\tau) \sum_{i\in B^e_1} x_{i}\frac{(W^e\vec{y})_{i}}{\vec{x}^TW^e\vec{y}} + \tau| B^e_1|
\\&=(1-n\tau)\frac{\sum_{i\in B^e_1} x_i (W^e\vec{y})_i}{\sum_{i\in G_1\backslash B^e_1} x_i (W^e\vec{y})_i+\sum_{i\in B^e_1} x_i (W^e\vec{y})_i} + \tau| B^e_1|
\\&\leq (1-n\tau)\frac{\sum_{i\in B^e_1} x_i (W^e\vec{y})_{i*}}{\sum_{i\in G_1\backslash B^e_1} x_i (W^e\vec{y})_i+\sum_{i\in B^e_1}
x_i (W^e\vec{y})_{i*}} + \tau| B^e_1| \;\;\;\;\;\;\;\;\;\;\;\;\;(*)
\\&\leq (1-n\tau)\frac{\sum_{i\in B^e_1} x_i (W^e\vec{y})_{i*}}{\sum_{i\in G_1\backslash B^e_1} x_i (W^e\vec{y})_{i*}+\sum_{i\in B^e_1}
x_i (W^e\vec{y})_{i*}} + \tau| B^e_1| \\& = (1-n\tau) \frac{(W^e\vec{y})_{i*} \sum_{i \in B^e_1}x_i}{(W^e\vec{y})_{i*}\sum_i x_i}+
\tau| B^e_1|
\\&= (1-n\tau)\sum_{i\in B^e_1}x_{i}+ \tau| B^e_1|
\end{align*}
where inequality (*) is true because if $\frac{a}{b}<1$ then $\frac{a}{b}<\frac{a+c}{b+c}$ for all $a,b,c$ positive. Hence after we add noise $\vec{\delta}$ with $||\vec{\delta}||_{\infty} = \delta$, the resulting vector $(\vec{x}',\vec{y}')$ (which is the next generation frequency vector) will satisfy $\sum_{i\in B^e_1} x'_{i} \leq (1-n\tau)\sum_{i\in B^e_1}x_{i}+ \tau| B^e_1|+\delta| B^e_1|$.
By setting $S_t = \sum_{i\in B^e_1}x_{i}(t)$ it follows that $S_{t+1} \leq (1-n\tau)S_t + (\tau+\delta)| B^e_1|$ and also $S_0 \leq 1$. Therefore
$S_{t} \leq (\tau+\delta) |B^e_1| \frac{1 - (1-n\tau)^{t}}{n\tau}+(1-n\tau)^t$. By choosing $t = -\frac{\ln (2n)}{\ln (1 - n\tau)} \approx
\frac{\ln (2n)}{n \tau}$ it follows that $\sum_{i\in B^e_1}x_{i}(t)\leq \frac{(1+o(1))|B^e_1|+1/2}{n} \leq \frac{2|B^e_1|}{n}$ where we used the assumption that $\delta = o_n(\tau)$. The same argument holds for $B^e_2$.
\end{proof}
\subsection{Proof of Lemma \ref{lem:phasetransition}}
\begin{proof}
By Lemma \ref{lem:badallele} after $\frac{\ln (2n)}{n\tau}$ generations it follows that
\begin{equation}\label{eqbad} \sum_{i \in B^e_1} x_{i}(t) + \sum_{j \in B^e_2} y_{i}(t) \leq  \frac{2|B^e|}{n}
\end{equation}
We consider the average fitness function $\vec{x}^TW^e\vec{y}$ which is not increasing (as has already been mentioned). Let $\vec{\tau} = \tau \cdot (1,...,1)^T$, $(\tilde{\vec{x}},\tilde{\vec{y}}) = f(\vec{x},\vec{y})$ and $(\hat{\vec{x}},\hat{\vec{y}}) = g(\vec{x},\vec{y})$ with fitness matrix $W^e$ and also denote by $(\vec{x}',\vec{y}')$ the resulting vector after noise $\vec{\delta}$ is added.
It is easy to observe that $$\tilde{\vec{x}}^TW^e\tilde{\vec{y}} = (1-n\tau)^2\hat{\vec{x}}^TW^e\hat{\vec{y}}+(1-n\tau)\hat{\vec{x}}^TW^e \vec{\tau} + (1-n\tau)\vec{\tau}^TW^e\hat{\vec{y}}+ \vec{\tau}^TW^e\vec{\tau}$$ and also that $$\vec{x}'^T W^e\vec{y}' \geq \tilde{x}^TW^e\tilde{y}-2n\delta W^e_{\max} \geq \tilde{x}^TW^e\tilde{y}\left(1-O\left(2n\delta\frac{W_{max}}{W_{min}}\right)\right) = (1-o_{n\tau}(1))\tilde{x}^TW\tilde{y}$$ where $W_{\max} = \max_e W^e_{\max}$. Under the assumption \ref{ass:avfitness} we have the following upper bounds:
\begin{itemize}
\item $\hat{\vec{x}}^TW^e \vec{\tau} \geq (1+\beta)n\tau \left(1 - \frac{2|B^e_1|}{n}\right)$ and $\hat{\vec{\tau}}^TW^e \hat{\vec{y}} \geq (1+\beta)n\tau \left(1 - \frac{2|B^e_2|}{n}\right)$.
\item $\vec{\tau}^TW^e \vec{\tau} \geq (n\tau)^2 (1+\beta)\left(1 - \frac{|B^e|}{n}\right) \geq (1+\beta)\left(1-\frac{2|B^e|}{n}\right)^2n^2\tau^2$.
\end{itemize}
First assume that $\vec{x}^TW^e\vec{y} \leq 1+\frac{\beta}{2}$. We get the following system of inequalities:
\begin{align*}
\frac{\vec{x}'^TW^e\vec{y}'}{\vec{x}^TW^e\vec{y}}&\geq (1-o_{n\tau}(1))\frac{\tilde{\vec{x}}^TW^e\tilde{\vec{y}}}{\vec{x}^TW^e\vec{y}}
\\&\geq (1-o_{n\tau}(1)) \left((1-n\tau)^2\frac{\hat{\vec{x}}^TW^e\hat{\vec{y}}}{\vec{x}^TW^e\vec{y}} + 2(1-n\tau)n\tau\left(1 - \frac{2|B^e|}{n}\right)\frac{(1+\beta)}{\vec{x}^TW^e\vec{y}}+\frac{(1+\beta)}{\vec{x}^TW^e\vec{y}}\left(1-\frac{2|B^e|}{n}\right)^2n^2\tau^2\right)
\\&\geq (1-o_{n\tau}(1))\left ((1-n\tau)^2 + 2(1-n\tau)n\tau\left(1 - \frac{2|B^e|}{n}\right)\left(1+\frac{\beta}{2+\beta}\right)+\left(1+\frac{\beta}{2+\beta}\right)\left(1-\frac{2|B^e|}{n}\right)^2n^2\tau^2\right)
\\&\geq (1-o_{n\tau}(1))\left ( 1 + n\tau \left(\frac{2\beta}{2+\beta} - \frac{6|B^e|}{n} - \frac{2\beta}{2+\beta}n\tau \right)\right)
\\&\geq 1+ n\tau \left(\frac{\beta}{2+\beta}\right)
\end{align*}
The second inequality comes from the fact that $\hat{\vec{x}}^TW^e\hat{\vec{y}} \geq \vec{x}^TW^e\vec{y}$ (the average fitness is increasing for the no mutation setting) and also since $\vec{x}^TW^e\vec{y} \leq 1+\frac{\beta}{2}$. The third and the fourth inequality use the fact that $|B^e| \ll n \beta$ and $\tau \leq \frac{\beta}{16n}$. Therefore, the fitness increases in the next generation for the mutation setting as long as the current fitness $\vec{x}^TW^e\vec{y} \leq 1 + \frac{\beta}{2}$ with a factor of $1+n\tau \frac{\beta}{2+\beta}$ (i). Hence the time we need to reach the value of $1$ for the average fitness is  $\frac{ 2\ln \frac{1}{h}}{ n\tau\frac{\beta}{2+\beta}}$ which is dominated by $t_1 = \frac{\ln (2n)}{n \tau}$. Therefore the total loss factor is at most $\frac{1}{d} = h^{t_1}$, namely $d = \left(\frac{1}{h}\right)^{t_1}$. Let $t_2$ be the time for the average fitness to reach $1+\frac{\beta}{4}$ (as long as it has already reached 1), thus $t_2 = \frac{2}{n\tau}$ which is dominated by $t_1$.
By similar argument, let's now assume that $\vec{x}^TW^e\vec{y} \geq 1 + \frac{\beta}{2}$ then
\begin{align*}
\frac{\vec{x}'^TW^e\vec{y}'}{\vec{x}^TW^e\vec{y}} &\geq (1-o_{n\tau}(1))\left(\frac{\tilde{\vec{x}}^TW^e\tilde{\vec{y}}}{\vec{x}^TW^e\vec{y}}\right)
\\&\geq (1-o_{n\tau}(1))\left((1-n\tau)^2\frac{\hat{\vec{x}}^TW^e\hat{\vec{y}}}{\vec{x}^TW^e\vec{y}} + 2(1-n\tau)n\tau\left(1 - \frac{2|B^e|}{n}\right)\frac{(1+\beta)}{\vec{x}^TW^e\vec{y}}+\frac{(1+\beta)}{\vec{x}^TW^e\vec{y}}\left(1-\frac{2|B^e|}{n}\right)^2n^2\tau^2\right)
\\&\geq 1 - 2n\tau
\end{align*}
Hence $\vec{x}'^TW^e\vec{y}' \geq (1 - 2n\tau)(1+ \frac{\beta}{2})$, namely $\vec{x}'^TW^e\vec{y}' \geq 1+ \frac{\beta}{4}$ (ii) for $\tau < \frac{\beta}{16n}$. Therefore as long as the fitness surpasses $1+\frac{\beta}{4}$, it never goes below $1+\frac{\beta}{4}$ (conditioned on the fact you remain at the same environment). This is true from claims (i), (ii). When the fitness is at most $1+\frac{\beta}{2}$, it increases in the next generation and when it is greater than $1+\frac{\beta}{2}$, it remains at least $1+\frac{\beta}{4}$ in the next generation.
To finish the proof we compute the times. The time $t_3$ to have a total gain factor of at least $d$, will be such that $(1+\frac{\beta}{4})^{t_3} = \frac{1}{h^{t_1}}$. Hence $t_3 = t_1 \frac{2\ln \frac{1}{h}}{\beta}$. By setting $T_{thr} = \frac{6 \ln (2n)}{ n\tau\beta W_{\min} } >  \frac{6 \ln \frac{1}{h} \ln (2n)}{n\tau\beta}> 3t_3 > t_1+t_2+t_3$ the proof finishes.
\end{proof}
\subsection{Proof of Theorem \ref{thm:potential}}
\begin{proof} We first prove the results for rational $\tau$; let $\tau = \nfrac{\kappa}{\lambda}$. We use the theorem of Baum and Eagon~\cite{BE66}. Let
\[L(\vec{x},\vec{y}) = (\vec{x}^T W\vec{y})^{\lambda-m\kappa} \prod_i x_i^{\kappa}\prod_i y_i^{\kappa}.\]
Then
\[x_{i}\frac{\partial L}{\partial x_{i}} = 2\kappa L + \frac{2x_i(W\vec{y})_i(\lambda -m\kappa)L}{\vec{x}^T W\vec{y}}.\]
It follows that
\begin{align*}
\frac{x_{i}\frac{\partial L}{\partial x_{i}}}{\sum_{i} x_{i}\frac{\partial L}{\partial x_{i}}}
&= \frac{2\kappa L + \frac{2x_i (W\vec{y})_i(\lambda-m\kappa)L}{\vec{x}^T W\vec{y}}}{2m\kappa L+2(\lambda - m\kappa)L}
\\
& = \frac{2\kappa L}{2 \lambda L}+ \frac{2L(\lambda-m\kappa)x_i (W\vec{y})_i}{2 \lambda L \vec{x}^T W\vec{y}}
\\&= (1-n\tau)x_{i} \frac{(W\vec{y})_{i}}{\vec{x}^T W\vec{y}}+\tau
\end{align*} where the first equality comes from the fact that $\sum _{i=1}^n x_i (W\vec{y})_i = \vec{x}^T W\vec{y}$. The same is true for $y_i \frac{\partial L}{\partial y_i}$. Since $L$ is a homogeneous polynomial of degree $2\lambda$, from Theorem~\ref{thm:inequality} we get that $L$ is strictly increasing along the trajectories, namely $$L(f(\vec{x},\vec{y})) > L(\vec{x},\vec{y})$$ unless $(\vec{x},\vec{y})$ is a fixed point ($f$ is the update rule of the dynamics, see also \ref{eq.fs}). So $P(\vec{x},\vec{y}) = L^{\nfrac{1}{\kappa}}(\vec{x},\vec{y})$ is a potential function for the dynamics.

To prove the result for irrational $\tau$, we just have to see that the proof of~\cite{BE66} holds for all homogeneous polynomials with degree $d$, even irrational.\\\\
To finish the proof let $\Omega \subset \Delta$ be the set of limit points of an orbit $\vec{z}(t) = (\vec{x}(t),\vec{y}(t))$ (frequencies at time $t$ for $t \in \mathbb{N}$). $P(\vec{z}(t))$ is increasing with respect to time $t$ by above and so, because $P$ is bounded on $\Delta$, $P(\vec{z}(t))$ converges as $t\to \infty$  to $P^{*}= \sup_t\{P(\vec{z}(t))\}$. By continuity of $P$ we get that $P(\vec{v})= \lim_{t\to\infty} P(\vec{z}(t)) = P^*$ for all $\vec{v} \in \Omega$. So $P$ is constant on $\Omega$. Also $\vec{v}(t) = \lim_{k \to \infty} \vec{z}(t_k + t)$ as $k \to \infty $ for some sequence of times $\{t_i\}$ and so $\vec{v}(t)$ lies in $\Omega$, i.e. $\Omega$ is invariant. Thus, if $\vec{v} \equiv \vec{v}(0) \in \Omega$ the orbit $\vec{v}(t)$ lies in $\Omega$ and so $P(\vec{v}(t)) = P^*$ on the orbit. But $P$ is strictly increasing except on equilibrium orbits and so $\Omega$ consists entirely of fixed points.
\end{proof}